\documentclass[onecolumn, a4paper, 11pt, 
  accepted=2025-02-26]{quantumarticle}
\pdfoutput=1

\usepackage{authblk}
\usepackage[dvipsnames]{xcolor}

\usepackage{amssymb}
\usepackage{amsmath}
\usepackage{amsfonts}
\usepackage{amsthm}
\allowdisplaybreaks[1]

\usepackage{enumerate}
\usepackage{indentfirst}

\usepackage[backref=page,colorlinks,linkcolor=Mulberry,anchorcolor=blue,citecolor=Mulberry]{hyperref}

\usepackage[english]{babel}
\usepackage[T1]{fontenc}

\usepackage{tikz}
\usetikzlibrary{quantikz2}

\usepackage[algoruled]{algorithm2e}

\usepackage{cleveref}
\crefname{ineq}{inequality}{inequalities}
\creflabelformat{ineq}{#2{\upshape(#1)}#3}
\crefname{fact}{fact}{facts}
\crefname{equation}{equation}{equations}
\crefname{algorithm}{algorithm}{algorithm}
\crefname{remark}{remark}{remarks}
\crefname{conjecture}{conjecture}{conjectures}

\usepackage{thmtools}
\declaretheorem[style=plain,numberwithin=section]{theorem}
\declaretheorem[style=plain,numberlike=theorem]{lemma,proposition,corollary}
\declaretheorem[style=remark,numberlike=theorem]{remark}
\declaretheorem[style=plain,numberlike=theorem]{definition,conjecture,fact}

\definecolor{darkgreen}{rgb}{0.1, 0.6, 0.1}

\renewcommand{\texttt}[1]{\begingroup \ttfamily \hyphenchar\font=`\- #1\endgroup}

\newcommand{\stateBQP}{\textnormal{\textsf{stateBQP}}\xspace}

\newcommand{\statePreciseBQP}{\textnormal{\textsf{statePreciseBQP}}\xspace}
\newcommand{\stateQMA}{\textnormal{\textsf{stateQMA}}\xspace}
\newcommand{\stateQCMA}{\textnormal{\textsf{stateQCMA}}\xspace}
\newcommand{\statePreciseQMA}{\textnormal{\textsf{statePreciseQMA}}\xspace}
\newcommand{\statePreciseQCMA}{\textnormal{\textsf{statePreciseQCMA}}\xspace}
\newcommand{\statePSPACE}{\textnormal{\textsf{statePSPACE}}\xspace}
\newcommand{\stateQIP}{\textnormal{\textsf{stateQIP}}\xspace}
\newcommand{\Ptime}{\textnormal{\textsf{P}}\xspace}

\newcommand{\BQP}{\textnormal{\textsf{BQP}}\xspace}

\newcommand{\QSZK}{\textnormal{\textsf{QSZK}}\xspace}
\newcommand{\NP}{\textnormal{\textsf{NP}}\xspace}
\newcommand{\IP}{\textnormal{\textsf{IP}}\xspace}

\newcommand{\QCMA}{\textnormal{\textsf{QCMA}}\xspace}
\newcommand{\QMA}{\textnormal{\textsf{QMA}}\xspace}
\newcommand{\UQMA}{\textnormal{\textsf{UQMA}}\xspace}
\newcommand{\PreciseUQMA}{\textnormal{\textsf{PreciseUQMA}}\xspace}
\newcommand{\QIP}{\textnormal{\textsf{QIP}}\xspace}
\newcommand{\PreciseQMA}{\textnormal{\textsf{PreciseQMA}}\xspace}
\newcommand{\PP}{\textnormal{\textsf{PP}}\xspace}

\newcommand{\PSPACE}{\textnormal{\textsf{PSPACE}}\xspace}
\newcommand{\NPSPACE}{\textnormal{\textsf{NPSPACE}}\xspace}
\newcommand{\BQPSPACE}{\textnormal{\textsf{BQPSPACE}}\xspace}
\newcommand{\UnitaryBQPSPACE}{\mathsf{BQ_{U}PSPACE}}
\newcommand{\stateUPSPACE}{\mathsf{state_{U}PSPACE}}
\newcommand{\stateQMAPSPACEoff}{\mathsf{stateQMA_{U}PSPACE}^\mathsf{off}}

\renewcommand{\bra}[1]{\langle #1|}
\renewcommand{\ket}[1]{| #1 \rangle}

\newcommand{\ketbra}[2]{\left| #1 \right\rangle \left\langle #2 \right|}
\newcommand{\innerprod}[2]{\langle #1 | #2 \rangle}

\newcommand{\rank}{\mathrm{rank}}
\newcommand{\Out}{\mathrm{Out}}

\newcommand{\Tr}{\mathrm{Tr}}
\DeclareMathOperator\erf{erf}

\newcommand{\yes}{\mathrm{yes}}
\newcommand{\no}{\mathrm{no}}

\newcommand{\F}{\mathrm{F}}
\newcommand{\td}{\mathrm{td}}
\newcommand{\Img}{\mathrm{Img}}

\newcommand{\sign}{\mathrm{sgn}}
\renewcommand{\path}{\mathrm{path}}

\newcommand{\Toffoli}{\textnormal{\textsc{Toffoli}}\xspace}
\newcommand{\Hadamard}{\textnormal{\textsc{Hadamard}}\xspace}

\newcommand{\T}{\textnormal{\textsc{T}}\xspace}
\newcommand{\CNOT}{\textnormal{\textsc{CNOT}}\xspace}

\renewcommand{\Pr}[1]{\mathrm{Pr}\left[#1\right]}

\newcommand{\binset}{\{0,1\}}

\newcommand{\poly}{\mathrm{poly}}
\newcommand{\negl}{\mathrm{negl}}
\newcommand{\In}{\mathrm{in}}
\renewcommand{\Out}{\mathrm{out}}

\newcommand{\bbC}{\mathbb{C}}

\newcommand{\bbR}{\mathbb{R}}
\newcommand{\bbN}{\mathbb{N}}

\newcommand{\calA}{\mathcal{A}}
\newcommand{\calG}{\mathcal{G}}
\newcommand{\calH}{\mathcal{H}}
\newcommand{\calL}{\mathcal{L}}
\newcommand{\calS}{\mathcal{S}}

\newcommand{\sfC}{\mathsf{C}}

\newcommand{\sfO}{\mathsf{O}}
\newcommand{\sfR}{\mathsf{R}}
\newcommand{\sfS}{\mathsf{S}}
\newcommand{\sfW}{\mathsf{W}}

\begin{document}
\setlength{\abovedisplayskip}{6pt}
\setlength{\belowdisplayskip}{6pt}

\title{Quantum Merlin-Arthur proof systems for synthesizing quantum states}
\author[2,1]{Hugo Delavenne}
\author[1]{François Le Gall}
\author[1]{Yupan Liu}
\author[1]{Masayuki Miyamoto}

\affil[1]{Graduate School of Mathematics, Nagoya University}
\affil[2]{ENS Paris-Saclay, Université Paris-Saclay}

\maketitle
\begin{abstract}
    Complexity theory typically focuses on the difficulty of solving computational problems with classical inputs and outputs, even with a quantum computer. In the quantum world, it is natural to apply a different notion of complexity, namely the complexity of synthesizing quantum states. We investigate a state-synthesizing counterpart of the class \NP{}, referred to as \stateQMA{}, which is concerned with preparing certain quantum states through a polynomial-time quantum verifier with the aid of a single quantum message from an all-powerful but untrusted prover. This is a subclass of the class \stateQIP{} recently introduced by \hyperlink{cite.RY22}{Rosenthal and Yuen (ITCS 2022)}, which permits polynomially many interactions between the prover and the verifier. 
    Our main result consists of error reduction of this class and its variants with an exponentially small gap or bounded space, as well as how this class relates to other fundamental state-synthesizing classes, i.e., states generated by uniform polynomial-time quantum circuits (\stateBQP{}) and space-uniform polynomial-space quantum circuits (\statePSPACE{}). 
    Furthermore, we establish that the family of \UQMA{} witnesses, considered as one of the most natural candidates for \stateQMA{} containments, is in \stateQMA{}. Additionally, we demonstrate that \stateQCMA{} achieves perfect completeness. 
\end{abstract}    

\section{Introduction}

Classical and quantum complexity theory typically concentrates on the computational difficulty of solving problems with \textit{classical} inputs and outputs. However, quantum computers have the ability to handle not only classical problems, but also quantum tasks, such as synthesizing quantum states. The most famous example is preparing ground states of a physical system~\cite{Lloyd96, PW09}, which even dates back to Feynman's original ideas~\cite{Feynman82}. Analogous tasks are also commonplace in quantum cryptography and generalized notions of the pseudorandomness, such as quantum money~\cite{Aaronson09} and pseudorandom quantum states~\cite{JLS18}. This motivates the study of complexity of synthesizing quantum states. 

In~\cite{Aaronson16}, Aaronson investigated the concept on quantum state complexity, leading to the \textit{state synthesizing problem}. This problem involves generating a quantum state $\rho$ from the all-zero state based on a quantum circuit with a succinct description acting on $n$ qubits (with the depth up to exponential). The resulting state $\rho$ is supposed to be close to the designated \textit{target state} $\ket{\psi}$.\footnote{We measure the closeness between $\rho$ and $\ket{\psi}$ by the trace distance, see \Cref{subsec:quantum-distances} for the definition.} This problem is solvable in (quantum) polynomial space (\PSPACE{}), i.e., a quantum computer running in exponential time but using polynomially many (possibly ancillary) qubits can generate a state that well approximates the target state. 

Quantum computers are seemingly not capable of solving any \PSPACE{} problem in polynomial time, while \textit{polynomially many} messages interactive protocols with the help of an all-powerful and \textit{untrusted} prover (known as \textit{interactive proofs}, \IP{}) captures the full computational power of polynomial-space computation, referred to as the celebrated $\IP=\PSPACE$ theorem~\cite{LFKN92,Shamir92}.
A recent line of works~\cite{RY22,MY23} initializes the study on the \textit{interactive} state synthesizing problem. Rosenthal and Yuen~\cite{RY22} denote the polynomial-space-preparable state families as \statePSPACE{}\footnote{The definition of \statePSPACE{} is a bit subtle: although all quantum states can be well-approximated by an exponentially long gate sequence owing to the Solovay-Kitaev theorem~\cite{Kitaev97}, this exponential gate sequence is not necessarily \textit{space-uniform}.} and show that such state families are preparable by \textit{polynomial-message} interactive synthesis protocols, which belongs to the class \stateQIP{}.  
Recently, Metger and Yuen~\cite{MY23} managed to prove the equivalence between state families that are preparable using polynomial space and those generated by interactive synthesis protocols, that is, $\stateQIP=\statePSPACE$, which is the state-synthesizing counterpart of the $\IP=\PSPACE$ theorem (and its quantum analogue~\cite{JJUW11}). 
Very recently (after our work was released), Rosenthal~\cite{Rosenthal23} strengthened the inclusion $\statePSPACE \subseteq \stateQIP$ by presenting a \textit{constant-message} \stateQIP{} protocol.
This result, combined with the other direction, can be considered as a parallelization technique counterpart for \stateQIP{}, highlighting that the behavior of \stateQIP{} aligns more closely with \QIP{} than \IP{}. 

However, there is currently a lack of fine-grained characterizations of computationally easier state families, viz., state families that are efficiently preparable (e.g., efficiently simulating view of quantum statistical zero-knowledge~\cite{Wat02}), or state families that are synthesizable via simply one-message interactive protocols (e.g., efficient verification of pure quantum states in the adversarial scenario~\cite{ZH19}). This opens up opportunities for our main results. 

\subsection{Main results}

In this work, we are particularly interested in state families that are preparable by a \textit{one-message} protocol, denoted as \stateQMA{}, which is obviously a subclass of \stateQIP{}. 
Let us first define \stateQMA{} informally (see \Cref{subsec:stateQMA-and-more} for a formal definition). 
A state family in \stateQMA{} is a family $\{\ket{\psi_x}\}_{x\in\calL}$ indexed by \textit{binary strings} in the corresponding language $\calL$ such that there is a verifier that has the following properties, verifying whether the target state $\ket{\psi_x}$ corresponding to a given input $x\in \calL$ can be well-approximated by the resulting state of this verifier. 
The verifier's computation, which is a polynomial-size \textit{unitary} quantum circuit,\footnote{In particular, extending to general quantum circuits does not change the class \stateQMA{} owing to the principle of deferred measurement. However, such extensions do not immediately work for the space-bounded counterpart, namely $\stateQMAPSPACEoff$, see \Cref{remark:state-synthesizing-with-unitary}.} takes a quantum-proof state $\ket{w}$ (with no limitations on the preparation method) and ancillary qubits in the state $\ket{\bar{0}}$ as input. After performing the verification circuit, a designated output qubit will be measured on the computational basis, and the verifier accepts if the measurement outcome is $1$. If the verifier accepts, the verification circuit has prepared the resulting state $\rho_{x,w}$ on the remaining qubits that is a good approximation of the target state $\ket{\psi_x}$ (if the verifier rejects, the resulting state could be anything). The acceptance probability is called \textit{the success probability} for approximately preparing $\ket{\psi_x}$. 

More precisely, the state family is in the class $\stateQMA_{\delta}[c,s]$ for some $0 \leq s < c \leq 1$ and $\delta \geq 0$, if the resulting state $\rho_{x,w}$ is $\delta$-close to the target state $\ket{\psi_x}$ provided that the verifier accepts with probability at least $s$ (soundness condition); and additionally there exists a quantum witness that makes the verifier accept with probability at least $c$ (completeness condition). 

It is pretty evident that $\stateQMA_{\delta}[c,s] \subseteq \stateQMA_{\delta'}[c',s']$ if $c'\leq c$, $s'\geq s$ and $\delta' \geq \delta$. However, how crucially does $\stateQMA_{\delta}[c,s]$ depend on its parameters? For commonplace complexity classes, viz. \BQP{}, \QMA{}, \QIP{}, etc., the dependence on such parameters is very weak: the class remains the same so long as the completeness $c$ and soundness $s$ differ by at least some inverse polynomial of $n\coloneqq |x|$ where $x\in \calL$. This is known as \textit{error reduction}, which typically involves performing the verification circuit in parallel and taking the majority vote. 

However, error reduction for \stateQMA{} requires a more nuanced approach. A simple parallel repetition of the verification circuit ends with a tensor product of the resulting state that evidently differs from the original state family. 
Therefore, error reduction for \stateQMA{} does need to preserve not only the quantum witness state $\ket{w}$, but also the resulting state $\rho_{x,w}$, referred to as the \textit{doubly-preserving error reduction} (ending with negligible errors): 

\begin{theorem}[Doubly-preserving error reduction for \stateQMA{}, informal version of \Cref{thm:error-reduction-stateQMA}]
\label{thm:informal-error-reduction-stateQMA}
For any $c(n)-s(n) \geq 1/\poly(n)$ and $0 \leq c(n),s(n) \leq 1$, we have 
\[\stateQMA_{\delta}[c,s] \subseteq \stateQMA_{\delta}[1-\negl,\negl].\]
\end{theorem}

Nevertheless, applying \Cref{thm:informal-error-reduction-stateQMA} to a \textit{polynomial-space-bounded} variant of $\stateQMA_{\delta}[c,s]$, which we denote by $\stateQMAPSPACEoff$,\footnote{We emphasize that $\stateQMAPSPACEoff$ is not a state-synthesizing counterpart of the class \NPSPACE{}, where ``off'' indicates \textit{offline} access to the witness state. See \Cref{remark:online-vs-offline} for details.} will result in \textit{exponential} space. To address this, we generalize \Cref{thm:informal-error-reduction-stateQMA} in a manner that preserves the polynomial space complexity, as presented in \Cref{thm:informal-error-reduction-bounded-stateQMA}.
Here in the class $\stateQMAPSPACEoff$, the verifier's computation stays polynomially space-bounded but may take \textit{exponential time} and the gap between the completeness $c$ and the soundness $s$ is at least some inverse-exponential of $n\coloneqq |x|$.

\begin{theorem}[Doubly-preserving error reduction for $\stateQMAPSPACEoff$, informal version of \Cref{thm:error-reduction-bounded-stateQMA}]
\label{thm:informal-error-reduction-bounded-stateQMA}
For any $c(n)-s(n) \geq \exp(-\poly(n))$ and $0 \leq c(n),s(n) \leq 1$, we have
\[\stateQMAPSPACEoff_{\delta}[c,s] \subseteq \stateQMAPSPACEoff_{\delta}[1-\negl,\negl].\]
\end{theorem}

We note that \Cref{thm:informal-error-reduction-stateQMA} is a state-synthesizing counterpart of the witness-preserving error reduction for \QMA{}~\cite{MW05,NWZ09}. Likewise, \Cref{thm:informal-error-reduction-bounded-stateQMA} shares similarities with error reduction for unitary quantum computations~\cite{FKLMN16} in the context of synthesizing states. Along the line of Marriott and Watrous~\cite{MW05}, we demonstrate that logarithmic-size quantum witness states are useless for \stateQMA{} (this variant is referred to as $\stateQMA[\log]$), as shown in \Cref{corr:informal-stateQMAlog-in-stateBQP}. Here \stateBQP{} is defined as a subclass of \statePSPACE{} with only polynomially many gates. 

\begin{corollary}[Informal version of \Cref{thm:stateQMAlog-in-stateBQP}]
    \label{corr:informal-stateQMAlog-in-stateBQP}
    $\stateQMA_{\delta}[\log] = \stateBQP_{\delta}.$
\end{corollary}

Resembling the approach of Fefferman and Lin~\cite{FL16}, we demonstrate that a variant of \stateQMA{} that admits an exponentially small gap between completeness and soundness, known as \statePreciseQMA{}, is contained in \statePSPACE{}:

\begin{corollary}[Informal version of \Cref{thm:statePreciseQMA-in-statePSPACE}]
    \label{corr:informal-stateQMA-in-statePSPACE}
    $\statePreciseQMA_{\delta} \subseteq \statePSPACE_{\delta}$.
\end{corollary}

Surprisingly, Corollary \ref{corr:informal-stateQMA-in-statePSPACE} shows that the distance parameter $\delta$ remains \textit{unchanged}, while a similar \statePSPACE{} containment following from~\cite{MY23} will worsen the distance parameter $\delta$, namely $\stateQMA_{\delta} \subseteq \statePSPACE_{\delta+1/\poly}$. 

\vspace{1em}
Furthermore, we demonstrate a natural instance of \stateQMA{} that illustrates the state family of \UQMA{} witnesses, as described in \Cref{thm:informal-UQMA-witness-in-stateQMA}. It is noteworthy that we concentrate on \UQMA{} witnesses rather than \QMA{} witnesses because \stateQMA{} comprises \textit{state families} whereas \QMA{} witnesses typically correspond to \textit{a subspace}, such as the ground space $\calS_H$ of a local Hamiltonian $H$ with a dimension $\dim(\calS_H) > 1$. 

\begin{theorem}[\UQMA{} witness family is in \stateQMA{}, informal version of \Cref{thm:UQMA-witness-in-stateQMA-scaling}]
    \label{thm:informal-UQMA-witness-in-stateQMA}
    For any promise problem $\calL=(\calL_{\yes},\calL_{\no})$, we have
    \begin{itemize}
        \item If $\calL\in \UQMA$, the family of witness states $\{\ket{w_x}\}_{x\in \calL_{\yes}}$ corresponding to yes instances is in $\stateQMA_{1/\poly}$; 
        \item If $\calL\in \PreciseUQMA_{1-{1/\exp}}$ \textup{(\PreciseUQMA{} with $1/\exp$ completeness error)}, the family of witness states $\{\ket{w_x}\}_{x\in \calL_{\yes}}$ corresponding to yes instances is in $\statePreciseQMA_{1/\exp}$.
    \end{itemize}
\end{theorem}

In addition, we prove that \stateQCMA{}, which is a variant of \stateQMA{} in which only classical witnesses (i.e., binary strings) are considered for both completeness and soundness conditions, can achieve perfect completeness: 

\begin{theorem}[\stateQCMA{} achieves perfect completeness, informal version of \Cref{thm:stateQCMA-perfect-completeness}]
    \label{thm:informal-stateQCMA-eq-stateQCMA1}
    For any $c(n) - s(n) \geq 1/\poly(n)$ and $0 \leq c(n), s(n) \leq 1$, we have
    \[\stateQCMA_{\delta}[c,s] \subseteq \stateQCMA_{\delta}[1,s']\] 
    for some $s'$ such that $1-s'(n) \geq 1/\poly(n)$. 
\end{theorem}

This result is analogous to the $\QCMA=\QCMA_1$ theorem~\cite{JN11,JKNN12} for synthesizing quantum states. 
In addition, it is worth noting that \Cref{thm:informal-stateQCMA-eq-stateQCMA1} also straightforwardly extends to \statePreciseQCMA{}. 

\subsection{Proof techniques}

The proof of \Cref{thm:informal-error-reduction-stateQMA} and \Cref{thm:informal-error-reduction-bounded-stateQMA} employs the quantum linear algebra techniques developed by Gilyén, Su, Low, and Wiebe~\cite{GSLW19}, specifically the quantum singular value discrimination. 

\paragraph{Error reduction for \stateQMA{} by manipulating singular values.} To elaborate on the intuition, we begin by briefly reviewing the witness-preserving error reduction for \QMA{}~\cite{MW05,NWZ09}. Consider a \QMA{} verification circuit $V_x$ that takes a quantum witness state $\ket{w}$ (on the register $\sfW$) and ancillary qubits in the state $\ket{\bar{0}}$ as input. The corresponding acceptance probability is $\|\ket{1}\bra{1}_{\Out} V_x \ket{w}\ket{\bar{0}}\|_2^2$, which is equal to a quadratic form $\bra{w} M_x \ket{w}$ where the matrix $M_x\coloneqq \bra{\bar{0}} V_x^{\dagger} \ket{1}\bra{1}_{\Out} V_x \ket{\bar{0}}$. It is not hard to see the maximum acceptance probability of $V_x$ is the largest eigenvalue of $M_x$. 
We then view $M_x=\Pi_{\In}\Pi\Pi_{\In}$ as a product of Hermitian projectors $\Pi_{\In}$ and $\Pi$ where $\Pi_{\In}=I_{\sfW}\otimes\ket{\bar{0}}\bra{\bar{0}}$ and $\Pi=V_x^{\dagger} \ket{1}\bra{1}_{\Out} V_x$. Remarkably, there exists an orthogonal decomposition of the Hilbert space, which the projectors $\Pi_{\In}$ and $\Pi$ act on, into \textit{one-dimensional} and \textit{two-dimensional} common invariant subspaces. This elegant decomposition property is referred as to the Jordan lemma~\cite{Jordan75}.\footnote{See~\cite{Reg06} for the detailed statement of the Jordan lemma, as well as a simple proof.} Marriott and Watrous~\cite{MW05} then take advantage of the Jordan lemma and present error reduction for \QMA{} that preserves the quantum witness state. 

However, this error reduction technique does not automatically preserve the resulting state, as required in \stateQMA{}, we thus need a more sophisticated technique, namely the quantum singular value transformation~\cite{GSLW19}. This technique generalizes the qubitization technique introduced by Low and Chuang~\cite{LC19} that was inspired by the aforementioned decomposition property.
Moving on to the maximum acceptance probability of a \stateQMA{} verifier $V_x$, it corresponds to the square root of the largest singular value of the matrix $A_x=\Pi_{\Out} V_x \Pi_{\In}$ where $\Pi_{\Out}\coloneqq \ket{1}\bra{1}_{\Out}$ is the final measurement. In Section 3.2 of~\cite{GSLW19}, the authors extend the Jordan lemma to the singular value scenarios. In particular, $\Img(\Pi_{\In})$ and $\Img(\Pi_{\Out})$ can be decomposed into one-dimensional or two-dimensional common invariant subspaces. Now let us focus on the specific case of \stateQMA{}, we notice that the right singular vectors of $A_n$ correspond to the quantum witness state $\ket{w}$, and the left singular vectors correspond to the resulting state $\rho_{x,w}$. Therefore, we obtain \textit{doubly-preserving error reduction} for \stateQMA (\Cref{thm:informal-error-reduction-stateQMA}) by manipulating the singular values accordingly.\footnote{Concretely speaking, the analysis of error reduction based on majority votes essentially corresponds to obtaining tail bounds for the Binomial distribution. By leveraging the central limit theorem, it becomes sufficient to estimate tail bounds for the normal distribution, referred to as the error function $\erf(x)$. The approximation polynomials of the sign function in~\cite{LC17} then achieve this task. } 

It is noteworthy that \Cref{thm:informal-error-reduction-stateQMA} differs from Theorem 38 in~\cite{GSLW19} since our construction is based on the \textit{projected unitary encoding} (e.g., the presented matrix $A_n$) instead of the block-encoding. 
Furthermore, for establishing \Cref{thm:informal-error-reduction-bounded-stateQMA}, we make use of an \textit{exponential-degree} approximation polynomial of the sign function where all coefficients are computable in \PSPACE{} to within \textit{exponential precision}~\cite{MY23}. We additionally observe that the proof techniques in~\cite{MY23} can be straightforwardly adapted to \textit{projected unitary encodings} instead of the block-encodings originally utilized in their work. 

\paragraph{Applications of error reduction for \stateQMA{}.} Along the line of Theorem 3.13 in~\cite{MW05}, which states $\QMA[\log]=\BQP$, with \Cref{thm:informal-error-reduction-stateQMA}, it seems to straightforwardly make for \Cref{corr:informal-stateQMAlog-in-stateBQP}. 
Nevertheless, we need a careful analysis on the resulting state to achieve the statement. 
Specifically, utilizing the error reduction for \stateQMA{}, we begin with a verifier with completeness $1-2^{-p(n)}$ and soundness $2^{-p(n)}$ where $p$ is a polynomial of $n\coloneqq |x|$. Then we replace the short quantum witness state $\ket{w}$ with a completely mixed state $I_{\sfW}/\dim(I_{\sfW})$, which gives us a computation meeting the soundness condition such that the soundness $s$ is preserved and the gap between the completeness and the soundness shrinks to some inverse-polynomial of $n\coloneqq |x|$. Although the new resulting state $\rho_{I_{\sfW}}$ may greatly differ from $\rho_{x,w}$, the definition of \stateQMA{} guarantees that $\rho_{I_{\sfW}}$ is also close to the target state because the acceptance probability of the verifier with $I_{\sfW}$ is greater than the soundness $s$. This proof also easily extends to \Cref{corr:informal-stateQMA-in-statePSPACE} employing \Cref{thm:informal-error-reduction-bounded-stateQMA}. 
In addition, it is noteworthy that \stateBQP{} achieves perfect completeness with a worsening distance parameter $\delta'$. By incorporating error techniques for both \stateBQP{} and $\stateUPSPACE$, the difference between the new distance parameter $\delta'$ and the original one can be made \textit{exponentially small}. 
Furthermore, we remark that \stateBQP{} is contained in \stateQMA{}.\footnote{See \Cref{subsec:stateQMA-and-more} (in particular, \Cref{prop:stateBQP-in-stateQMA}) for a detailed elaboration. } 
We therefore complete the other direction in \Cref{corr:informal-stateQMAlog-in-stateBQP}. 

\paragraph{\UQMA{} witness family is in \stateQMA{}.} To establish \Cref{thm:informal-UQMA-witness-in-stateQMA}, we start by employing a weighted circuit-to-Hamiltonian construction~\cite{BC18,BCNY19} to generate a ground state (also known as the history state) that provides a $\delta_1$-close approximation of the output state produced by a \UQMA{} verifier $V_x$ with maximum acceptance probability $1-\nu$. The resulting Hamiltonian $H^{(x)}$ possesses several crucial properties: the ground energy of $H^{(x)}$ is $\nu$, and the spectral gap satisfies $\Delta(H^{(x)}) \geq q(\nu \delta_1/T^3)$ for some polynomial $q$, where $T$ is the size of $V_x$. 

Next, we construct a $\stateQMA_{\delta}[c',s']$ verifier primarily based on the one-bit precision phase estimation~\cite{Kitaev95}, often referred to as Hadamard test~\cite{AJL09}. Through a fairly complicated analysis, we demonstrate that the promise gap $c'-s' \geq q'(v\delta_1 t/T^3)-2\epsilon$ for some polynomial $q'$ and the distance parameter $\delta\leq O(\max\{(1-\cos{\nu})/2,\delta_1,q(\nu \delta_1 t/T^3)\})$. Here, $\epsilon$ represents the implementation error of $\exp(-iH^{(x)}t)$ guaranteed by space-efficient Hamiltonian simulation~\cite[Theorem 12]{FL18}. We then achieve \Cref{thm:informal-UQMA-witness-in-stateQMA} by appropriately choosing parameters.

\paragraph{\stateQCMA{} achieves perfect completeness.} Our proof of \Cref{thm:informal-stateQCMA-eq-stateQCMA1} takes inspiration from \cite{JN11,JKNN12}, but it requires several modifications. Note that our concern in \stateQCMA{} is not only the maximum acceptance probability but also the resulting state after performing the verification circuit and the final measurement. To meet these requirements, we must choose a specific universal gateset $\calS$ such that $\calS$ can generate a dense subgroup of $\mathrm{SU}(2^n)$ where $n$ is the number of qubits utilized in the quantum circuit, and all quantum states generated by these gates in $\calS$ have rational entries. For this reason, we opt for the ``Pythagorean gateset''\footnote{See \Cref{remark:choices-of-gatesets} for the details.}~\cite{JN11,CK22}.
To ensure that the resulting state is indeed close to the target state, we slightly adjust the construction outlined in~\cite{JKNN12}. 

\subsection{Related works}

In addition to the state-synthesizing complexity classes explored in prior works~\cite{RY22,INNSY22,MY23} and the present study, there are other investigations focusing on the quantum state synthesizing problem from diverse perspectives, including cryptography and distributed computing. From a cryptographic standpoint, a recent work~\cite{CGK21} examines non-interactive zero-knowledge (NIZK) proofs for quantum states, with a relevance to the setting of \stateQCMA{}. Another concurrent study~\cite{CMS23} addresses zero-knowledge proofs for quantum states, considering scenarios that pertinent to both \stateQMA{} and \stateQIP{} scenarios. On the other hand, in the realm of distributed computing, another recent work~\cite{LGMN22} delves into quantum state synthesizing through the utilization of distributed quantum Merlin-Arthur protocols. 

\subsection{Discussion and open problems}

\paragraph{Reduction and completeness in state-synthesizing complexity theory.} In the context of state-synthesizing complexity theory, including prior works~\cite{INNSY22,MY23,RY22} and our own, the concepts of \textit{reduction} and \textit{completeness} have not been defined. However, these concepts hold significant importance in (quantum) complexity theory. The immediate challenge lies in appropriately defining these concepts, such as reduction, in a manner that ensures the resulting states exhibit reasonable behavior before and after the application of the reduction. 
Nevertheless, a very recent work~\cite{BEMPQY23} introduced notions of reduction and completeness within the context of \textit{unitary-synthesising} complexity theory. 

\paragraph{The computational power of \statePreciseQMA{}.} Although Corollary \ref{corr:informal-stateQMA-in-statePSPACE} establishes that a $\statePSPACE$ containment of $\statePreciseQMA$, the reverse inclusion, namely $\statePSPACE \subseteq \statePreciseQMA$, remains an open problem. The main challenge lies in adapting existing proof techniques that demonstrate $\PSPACE \subseteq \PreciseQMA$~\cite{FL16,FL18,Li22}, as these techniques heavily rely on notions of \textit{completeness} or \textit{reduction} for the class $\PSPACE$. 

\section{Preliminaries}

In this paper we assume the reader is familiar with the basic concepts and notions of quantum computation (see, e.g.,~\cite{NC10} for the reference). We begin by introducing some useful linear-algebraic notations. Let $M$ be an Hermitian matrix, and we denote the largest eigenvalue of $M$ as $\lambda_{\max}(M)$, as well as the largest singular value of $M$ as $\sigma_{\max}(M)$. Also, $I_k$ is the identity matrix of size $2^k \times 2^k$. In addition, we will employ three Schatten norms of matrices in this paper, which include: 
\begin{itemize}
    \item Trace norm $\|M\|_1\coloneqq \Tr \sqrt{M^{\dagger} M}$;
    \item Frobenius norm $\|M\|_2\coloneqq \sqrt{\Tr(M^{\dagger} M)}$; 
    \item Operator norm $\|M\|_{\infty}\coloneqq \sigma_{\max}(M) = \sqrt{\lambda_{\max}(M^{\dagger} M)}$. 
\end{itemize}

\subsection{Distances for quantum states}
\label{subsec:quantum-distances}

For any (in general mixed) quantum states $\rho_0$ and $\rho_1$, we measure the closeness between these states by the \textit{trace distance} $\td(\rho_0,\rho_1)\coloneqq \frac{1}{2}\|\rho_0-\rho_1\|_1$. 
Moreover, we will employ the following properties of the trace distance, which are explained in detail in Section 9.2.1 of~\cite{NC10}:
\begin{description}
    \item[Contractivity] For any quantum channel $\Phi$, we have $\td(\Phi(\rho_0),\Phi(\rho_1)) \leq \td(\rho_0,\rho_1)$; 
    \item[Convexity] For any non-negative coefficients $\{p_i\}_i$ such that $\sum\nolimits_{i} p_i=1$, we have \[\td\big( \rho_0, \sum\nolimits_i p_i \rho_1^{(i)} \big) \leq \sum\nolimits_i p_i \td\big(\rho_0,\rho_1^{(i)}\big).\] 
\end{description}

In addition, we utilize the (Uhlmann) \textit{fidelity} defined as \textit{fidelity} $\F(\rho_0,\rho_1)\coloneqq \Tr\sqrt{\sqrt{\rho_0} \rho_1 \sqrt{\rho_0}}$. 
Importantly, when considering a pure state $\rho_0\coloneqq \ket{\psi_0}\bra{\psi_0}$ and a mixed state $\rho_1$, the fidelity simplifies to $\F(\rho_0,\rho_1)=\sqrt{\bra{\psi_0} \rho_1 \ket{\psi_0}}$. 
Furthermore, the well-known Fuchs-van de Graaf inequality~\cite{FvdG99} states the following:
\begin{lemma}[Trace distance vs.~Fidelity,\cite{FvdG99}]
    \label{lemma:td-vs-fidelity}
    For any quantum states $\rho_0$ and $\rho_1$, \[1-\F(\rho_0,\rho_1) \leq \td(\rho_0,\rho_1) \leq \sqrt{1-\F^2(\rho_0,\rho_1)}.\]
\end{lemma}

\subsection{Gatesets matter for synthesizing quantum states}

A \textit{gateset} is a finite set of unitary matrices each of which acts on a finite-dimensional quantum system.
A gateset $\calS$ is \textit{universal} if the subgroup generated by $\cal{S}$ is dense in $\mathrm{SU}(2^n)$ for large enough $n$. 
We mainly use the common universal gateset $\{\CNOT, \Hadamard, \T\}$ for convenience, and further note that complexity classes remain unchanged for all reasonable choices of gatesets that all entries are algebraic numbers,\footnote{For a comprehensive explanation of the conditions on choosing gatesets, see Theorem 2.10 in~\cite{Kuperberg15}.} owing to the Solovay-Kitaev theorem~\cite{Kitaev97} and its space-efficient variant~\cite{vMW12}. 

Additionally, to achieve perfect completeness, we require a particular \textit{``Pythagorean'' gateset} in \Cref{sec:stateQCMA}, introduced by Jordan and Nagaj\cite{JN11}, which consists of \CNOT{} and ``Pythagorean'' gates
$\frac{1}{5}
\begin{psmallmatrix}
4 & -3 \\
3 & 4 
\end{psmallmatrix}$ and 
$\frac{1}{5}
\begin{psmallmatrix}
4 & 3i \\
3i & 4 
\end{psmallmatrix}$.\footnote{An analogous gateset is also utilized in \cite{CK22} with a slightly different ``Pythagorean'' gate.} 

We end with \Cref{lemma:changing-gateset-errs} which states changing the gateset of a family of quantum circuits that prepares a certain family of quantum states will make the resulting state negligibly far from the target state. \Cref{lemma:changing-gateset-errs} highlights the robustness of the state-synthesizing procedure with respect to the choice of gateset. For convenience, we will define $\psi\coloneqq \ket{\psi}\bra{\psi}$. 

\begin{lemma}[Changing gateset worsens the distance parameter]
\label{lemma:changing-gateset-errs}
Consider a family of (verification) circuits $\{Q_x\}_{x\in\calL}$ that prepares the corresponding family of resulting states $\{\rho_x\}_{x\in\calL}$.\footnote{In particular, we first perform the (verification) circuit $Q_x$ on the input state, which is not necessarily the all-zero state, then we measure the designated output qubit. If the measurement outcome is $1$, we obtain the resulting state corresponding to $Q_x$ on the remaining qubits.}
Then there exists a circuit family $\{Q'_x\}_{x\in\calL}$ that prepares $\{\rho'_x\}_{x\in\calL}$, using a designated universal gateset $\calG$ that is closed under adjoint, such that $\td(\rho_x,\rho'_x) \leq \exp(-\poly(n))$ where $n\coloneqq |x|$. 
In particular, the deterministic algorithm that constructs $\{Q'_x\}_{x\in\calL}$ runs in polynomial time and polynomial space.
\end{lemma}

\begin{proof}
Consider a quantum channel $\Phi(\rho)\coloneqq \frac{\Pi_{\Out} \rho \Pi_{\Out}}{\Tr\left( \Pi_{\Out} \rho \Pi_{\Out} \right)}$ where $\Pi_{\Out}\coloneqq \ket{1}\bra{1}_{\Out}$ that post-selects the output qubit to be $1$, and let $\ket{w}$ be a (quantum) witness,\footnote{If $Q_x$ and $Q'_x$ are not verification circuits, then this witness state is simply the all-zero state.} we have derived that
\begin{align*}
    \td(\rho_x,\rho'_x)&=\td\left(\Phi\left(Q_x(\ket{w}\bra{w}\otimes\ket{\bar{0}}\bra{\bar{0}})Q_x^\dagger\right), \Phi\left(Q'_x(\ket{w}\bra{w}\otimes\ket{\bar{0}}\bra{\bar{0}})(Q'_x)^{\dagger}\right)\right)\\
    &\leq \td\left(Q_x(\ket{w}\bra{w}\otimes\ket{\bar{0}}\bra{\bar{0}})Q^\dagger_x, Q'_x(\ket{w}\bra{w}\otimes\ket{\bar{0}}\bra{\bar{0}})(Q'_x)^{\dagger}\right)\\
    &= \sqrt{1-\left|\bra{w}\bra{\bar{0}}Q_x^{\dagger}Q'_x\ket{w}\ket{\bar{0}}\right|^2}\\
    &\leq \|Q_x\ket{w}\ket{\bar{0}}-Q'_x\ket{w}\ket{\bar{0}}\|^2_2\\
    &= \|Q_x-Q'_x\|^2_{\infty} \|\ket{w}\ket{\bar{0}}\|^2_2,
\end{align*}
with the second line owing to the contractive property, and the fourth line due to the fact that $(\innerprod{\psi}{\phi}-1)^{\dagger}(\innerprod{\psi}{\phi}-1) \geq 0$ for any pure states $\ket{\psi}$ and $\ket{\phi}$.

Note that the gates of $\mathcal{G}$ are closed under adjoint. By a space-efficient Solovay-Kitaev theorem (i.e., Theorem 4.3 in~\cite{vMW12}), we can construct $Q'_x$ such that $\|Q_x-Q'_x\|_{\infty}\leq \epsilon$ by a deterministic algorithm running in time $\poly\log(1/\epsilon)$ and space $O(\log(1/\epsilon))$. By choosing $\epsilon^2=\exp(-\poly(n))$, we thus obtain $\td(\rho_x,\rho'_x) \leq \|Q_x-Q'_x\|^2_{\infty}\ \leq \exp(-\poly(n))$ as desired. 
\end{proof}

\subsection{Unique-witness \QMA{}}

\begin{definition}[Unique Quantum Merlin-Arthur, $\UQMA$, adapted from~\cite{ABOBS22,JKKSSZ12}]
    A promise problem $\calL=(\calL_{\yes},\calL_{\no})$ is in $\UQMA[c,r,s]$ if there exists a uniformly generated polynomial-size quantum circuit $\{V_x\}_{x\in \calL}$ having $k(n)$ qubits for the witness register and requiring $m(n)$ ancillary qubits initialized to $\ket{0^m}$, where $k$ and $m$ are polynomials of $n\coloneqq |x|$ such that 
    \begin{itemize}
        \item \textbf{\emph{Completeness}}. If $x \in \calL_{\yes}$, then there exists a witness $\ket{w}$ such that $\Pr{V_x \text{ accepts } \ket{w}} \geq c(n)$. Furthermore, for all states $\ket{\phi}$ orthogonal to $\ket{w}$, $\Pr{V_x \text{ accepts } \ket{\phi}} \leq r(n)$. 
        \item \textbf{\emph{Soundness}}. If $x \in \calL_{\no}$, then for all states $\ket{\phi}$, $\Pr{V_x \text{ accepts } \ket{\phi}} \leq s(n)$. 
    \end{itemize}
    Here, $c,r,s\colon \bbN \rightarrow [0,1]$ are efficiently computable functions such that $c(n) > \max\{r(n),s(n)\}$. 
\end{definition}

Following the witness-preserving error reduction for \QMA{} by Marriott and Watrous~\cite{MW05}, we know that \UQMA{} admits error reduction. 

\begin{theorem}[Error reduction for \UQMA{}, {\cite[Theorem 3.6]{JKKSSZ12}}]
\label{thm:UQMA-error-reduction}
Consider $x\in \calL$ where $\calL \in \UQMA$. 
Let $c,r,s\colon \bbN \rightarrow [0,1]$ be efficiently computable functions such that $c(n) > \max\{r(n),s(n)\}+1/q(n)$ for some polynomial $q$ of $n\coloneqq |x|$. 
Then for any polynomial $l(n)$, we have that 
\[\UQMA[c,r,s] \subseteq \UQMA\left[1-2^{-l},2^{-l},2^{-l}\right].\] 
The witness length of the new \UQMA{} verifier remains preserved.
\end{theorem}

Furthermore, we can similarly define a variant of $\UQMA$ that corresponds to a promise gap, specifically $c-\max\{r,s\}$, exponentially approaching zero. This variant is denoted as $\PreciseUQMA$. Additionally, we introduce the notations $\UQMA_{1-\nu}$ and $\PreciseUQMA_{1-\nu}$ when the completeness error of these classes is $\nu$:
\begin{align*}
    \UQMA_{1-\nu}&\coloneqq \cup_{1-\nu-\max\{r(n),s(n)\} \geq 1/\poly(n)} \UQMA[1-\nu,r,s],\\
    \PreciseUQMA_{1-\nu}&\coloneqq \cup_{1-\nu-\max\{r(n),s(n)\} \geq \exp(-\poly(n))} \UQMA[1-\nu,r,s].
\end{align*}
In particular, if $\nu$ is exponentially small (negligible), we use the notations $\UQMA_{1-\negl}$ and $\PreciseUQMA_{1-\negl}$. 

\subsection{Tools from random walks on graphs}

In this subsection, we recapitulate several useful lemmas concerning random walks on an undirected path graph $G=(V,E)$ and their associated spectral properties. These lemmas have been employed in weighted circuit-to-Hamiltonian constructions, such as~\cite{BC18} and~\cite[Section 6]{BCNY19}. A random walk on such a graph corresponds to a reversible Markov chain $(\Omega,P,\pi)$ on the state space $\Omega=V$ with a stationary distribution $\pi$ and a transition matrix $P$. For detailed definitions of these terminologies, we refer to~\cite[Chapter 1]{LP17}. 

\begin{lemma}[Adapted from {\cite[Lemma 4]{BC18}}]
    \label{lemma:RW-on-path}
    Consider a $(T+1)$-node path graph $G=(V,E)$ and a strictly positive probability distribution $\pi$ on $V$.\footnote{A probability distribution $\pi$ is \textit{strictly positive} if all components $\pi_t$ are positive.} There exists a Markov chain $(V,P,\pi)$ with Metropolis transition probabilities $P$ that is reversible with respect to the given stationary distribution $\pi$ as below:     
    \begin{itemize}
        \item For $0 < t < T$, $P_{t,t+1}\coloneqq \frac{1}{4}\min\big\{1,\frac{\pi_{t+1}}{\pi_t}\big\}$, $P_{t+1,t}\coloneqq \frac{1}{4}\min\big\{ 1,\frac{\pi_{t-1}}{\pi_t} \big\}$, and $P_{t,t}\coloneqq 1-P_{t,t+1}-P_{t+1,t}$; 
        \item For $t=0$, $P_{0,1}\coloneqq \frac{1}{2} \min\big\{ 1,\frac{\pi_1}{\pi_0} \big\}$, $P_{0,-1}\coloneqq 0$, and $P_{0,0}\coloneqq 1-P_{0,1}$; 
        \item For $t=T$, $P_{T,T-1}\coloneqq \frac{1}{2}\min\big\{ 1,\frac{\pi_{T-1}}{\pi_T} \big\}$, $P_{T,T+1}\coloneqq 0$, and $P_{T,T}\coloneqq 1-P_{T,T-1}$. 
    \end{itemize}
    
    Furthermore, we have a corresponding stoquastic frustration-free\footnote{We say that a Hamiltonian is \textit{stoquastic} if all off-diagonal entries in the corresponding Hermitian matrix are real and non-positive. Additionally, we say that a Hamiltonian is \textit{frustration-free} if its ground energy is zero, and the ground states of the Hamiltonian $H=\sum_i H_i$ are also ground states of all local terms $H_i$.} $O(\log n)$-local Hamiltonian $H^{(T)}_{\path}\coloneqq \sum_{t=0}^T H_{\path}(t)$ such that the spectral gap $\Delta(H^{(T)}_{\path})$ is the same as the spectral gap $\Delta_P$.\footnote{Specifically, $\Delta(H^{(T)}_{\path})$ is the difference between the second smallest eigenvalue of $H^{(T)}_{\path}$ and the smallest eigenvalue, whereas $\Delta_P$ is the difference between the largest eigenvalue of $P$ and the second largest eigenvalue. } Here, the local term $H_{\path}(t)$ for any $0 \leq t \leq T$ is defined as: 
    \[H_{\path}(t) \coloneqq  (1-P_{t,t})\ket{t}\bra{t}
    - \tfrac{\sqrt{\pi_t}}{\sqrt{\pi_{t-1}}} P_{t,t-1} \ket{t}\bra{t-1}
    - \tfrac{\sqrt{\pi_{t-1}}}{\sqrt{\pi_t}} P_{t-1,t} \ket{t-1}\bra{t}\]
\end{lemma}

\begin{proof}
Our construction and the analysis closely parallel the proof of~\cite[Lemma 4]{BC18}. Through a direct calculation, we can establish that the matrix $P$ satisfies the condition of being a transition matrix: $\sum_{t=0}^T \pi_t P_{t,t'} = \pi_{t'}$. Additionally, the constructed Markov chain is reversible as the transition probabilities satisfy $\pi_t P_{t,t'} = \pi_{t'} P_{t',t}$ for all $0 \leq t,t' \leq T$.

According to the equivalent conditions for a reversible Markov chain (see~\cite{SJ89}), this construction gives rise to a symmetric matrix $A\coloneqq D^{1/2} P D^{-1/2}$, where $D\coloneqq {\rm diag}(\pi_0,\pi_1,\cdots,\pi_T)$. 
Since we obtain $A$ by applying a similarity transformation to $P$, all eigenvalues of $A$ coincide with those of $P$. Consequently, we arrive at $H^{(T)}_{\path}=I-A$ where $I$ is a $(T+1)\times (T+1)$ identity matrix. Notably, $H^{(T)}_{\path}$ is both stoquastic and frustration-free because $P$ is entry-wise non-negative and the largest eigenvalue of $P$ is $1$.
\end{proof}

Then we will provide a lower bound for the spectral gap $\Delta_P$ based on the \textit{conductance} of the underlying graph $G=(V,E)$, which characterizes the graph's connectivity. In particular, for any nonempty subset $S \subseteq V$, we define the conductance as $\Phi_{G}(S)\coloneqq \frac{1}{\min\{\pi(S),\pi(\bar{S})\}} \sum_{x \in S, y\in \bar{S}} \pi_x P_{x,y}$, and the conductance of the graph $G$ is denoted by $\Phi_{G}\coloneqq \min_{S \subseteq V} \Phi(S)$. 
The well-known Cheeger's inequality demonstrates a spectral gap lower bound as follows. 
\begin{lemma}[Cheeger's inequality, adapted from~{\cite[Lemma 3.3]{SJ89}}]
    \label{lemma:Cheeger-inequality}
    For a reversible Markov chain with the underlying (weighted) graph $G$, the spectral gap $\Delta_P$ of the transition matrix $P$ satisfies
    $\tfrac{1}{2} \Phi_G^2 \leq \Delta_P$. 
\end{lemma}

\subsection{Space-efficient Hamiltonian simulation}

We adopt the following definition of \textit{efficient encodings} of matrices adapted from~\cite{FL18}:
\begin{definition}[Efficient encodings, adapted from~{\cite[Definition 10]{FL18}}]
    Let $M$ be a $2^{k(n)}\times 2^{k(n)}$ matrix, and let $\calA$ be a classical algorithm (e.g., a Turing machine) specified using $n$ bits.
    We say that $\calA$ provides an efficient encoding of $M$ if, given an input $i \in \binset^{k(n)}$, $\calA$ outputs the indices and contents of the nonzero entries of the $i$-th row, using at most $\poly(n)$ time and $O(k)$ workspace (excluding the output size).
    Consequently, $M$ has at most $\poly(n)$ nonzero entries in each row.
\end{definition}

In the following scenario, we will specify a matrix $M$ in the input by providing an efficient encoding of $M$, where the size of the encoding is the input size denoted by $n$.

\vspace{1em}
We now proceed with a Hamiltonian simulation procedure achieving exponential precision with only polynomial space. The techniques come from \cite{BCCKS14,BCCKS15,BCK15}, with an observation on the space complexity made in~\cite{FL18}. 
\begin{lemma}[Adapted from~{\cite[Theorem 12]{FL18}}]
    \label{lemma:space-efficient-Hamiltonian-simulation}
    Consider the size-$n$ efficient encoding of a $2^{k(n)} \times 2^{k(n)}$ Hermitian matrix $H$ as an input, and treat this matrix $H$ as a Hamiltonian. Then the time evolution $\exp(-iHt)$ can be simulated using $\poly(n, k, \|H\|_{\infty}, t, \log(1/\epsilon))$ operations and $O(k + \log(t/\epsilon))$ space, where $\epsilon$ is the error in operator norm. 
\end{lemma}

\section{Basic state-synthesizing complexity classes}

In this section, we will define state-synthesizing complexity classes involved in this paper. 

\subsection{State-synthesizing complexity classes with bounded time or space}

We begin with two crucial notations on circuit families.
A family of quantum circuits $\{Q_x\}_{x\in \calL}$ associated with the language $\calL \subseteq \binset^*$ is \textit{uniformly generated} if there exists a deterministic polynomial-time Turing machine (and thus polynomially space-bounded) that on input $x \in \calL$ outputs the description of $Q_x$. 
Likewise, a family of quantum circuits $\{Q_x\}_{x\in \calL}$ is \textit{space-uniformly generated} if $Q_x$ acts on polynomially many qubits and there exists a deterministic polynomial-space Turing machine that on input $x \in \calL$ outputs the description of $Q_x$. 
Then we move on to the class that captures efficiently preparable quantum state families. 

\begin{definition}[$\stateBQP_{\delta}{[\gamma]}$]
\label{def:stateBQP}
A family of quantum states $\{\ket{\psi_x}\}_{x\in\calL}$ is in $\stateBQP_{\delta}[\gamma]$ if each $\ket{\psi_x}$ is an $n$-qubit state where $n\coloneqq |x|$, and there exists a uniformly generated quantum circuit family $\{Q_x\}_{x \in \calL}$, which takes no inputs, such that for all $x \in \calL$: 
\begin{itemize}
    \item The circuit $Q_x$ outputs a (in general mixed) quantum state $\rho_x$ satisfying $\td(\rho_x,\psi_x)\leq \delta(n)$ if $Q_x$ succeeds. Here, we define the success of circuit $Q_x$ as the measurement outcome of the designated output qubit being $1$. 
    \item The success probability of $Q_x$ is at least $\gamma$. 
\end{itemize}
\end{definition}

For convenience, we define $\stateBQP_{\delta}\coloneqq \stateBQP_{\delta}[2/3]$. Similarly, we define \stateBQP{} with exponentially small success probability in \Cref{def:statePreciseBQP}.
\begin{definition}[$\statePreciseBQP_{\delta}$]
\label{def:statePreciseBQP}
$\statePreciseBQP_{\delta} \coloneqq  \cup_{\gamma\geq \exp(-\poly(n))} \stateBQP_{\delta}[\gamma]$.
\end{definition}
Afterwards, we consider a state-synthesizing counterpart of unitary quantum polynomial space $\UnitaryBQPSPACE$ in~\cite{FL18}. In particular, we define \statePSPACE{} in terms of \textit{unitary quantum computation}, denoted as $\stateUPSPACE$. 

\begin{definition}[$\stateUPSPACE_{\delta}{[\gamma]}$]
\label{def:statePSPACE}
A family of quantum states $\{\ket{\psi_x}\}_{x\in\calL}$ is in $\stateUPSPACE_{\delta}[\gamma]$ if each $\ket{\psi_x}$ is an $n$-qubit state where $n\coloneqq |x|$, and there exists a space-uniform family of unitary quantum circuits $\{Q_x\}_{x \in \calL}$, which takes no inputs, such that for all $x \in \calL$:
\begin{itemize}
    \item The circuit $Q_x$ outputs a (in general mixed) quantum state $\rho_x$ satisfying $\td(\rho_x,\psi_x)\leq \delta(n)$ if $Q_x$ succeeds. Here, we define the success of circuit $Q_x$ as the measurement outcome of the designated output qubit being $1$. 
    \item The success probability of $Q_x$ is at least $\gamma$. 
\end{itemize}
\end{definition}

We denote $\stateUPSPACE_{\delta} \coloneqq  \stateUPSPACE_{\delta}[2/3]$. 
It is worth noting that there are three differences between our definition  $\stateUPSPACE$ and \statePSPACE{} as defined in~\cite{RY22,MY23}: 1) our definition only admits unitary quantum computation whereas \statePSPACE{} also permits general quantum circuits (see \Cref{remark:state-synthesizing-with-unitary}); 2) our definition does not assume perfect completeness, even though this property is always achievable for our definition with a worsening distance parameter; and 3) our definition merely deals with pure state families, while \statePSPACE{} in~\cite{MY23} also accommodates mixed state families. 

\begin{remark}[Synthesizing states by general quantum circuits]
    \label{remark:state-synthesizing-with-unitary}
    General quantum circuits allow not only unitary quantum gates but also intermediate measurements and resetting qubits. 
    Because of the principle of deferred measurement, admitting intermediate measurements does not make the time-efficient class more powerful. Recent developments on eliminating intermediate measurements in quantum logspace~\cite{FR21,GRZ21,GR22} improve this to the space-efficient scenario. Nevertheless, these new techniques in~\cite{FR21,GRZ21} do not automatically extend to state-synthesizing classes since they do not (even approximately) preserve the resulting state, and another result~\cite{GR22} only apply to a subclass of general quantum computation. 
\end{remark}

\paragraph{\stateBQP{} and $\stateUPSPACE$ achieve perfect completeness.}
It is worth noting that the family of states constructed by the simulator used in quantum statistical zero-knowledge~\cite{Wat02} is an instance of \stateBQP{}. Interestingly, one can assume that the corresponding simulator achieves perfect completeness with a worsening distance parameter $\delta'$.\footnote{To be specific, this fact is demonstrated in the \QSZK{}-hardness proof of the Quantum State Distinguishability Problem, as presented in Theorem 6 in~\cite{Wat02}.\label{footnote:QSZK-simulator}} To achieve this, one can simply replace the designated output qubit of the \QSZK{} simulator, or a \stateBQP{} circuit in general, with a qubit in state $\ket{1}$. This technique applies to any state-synthesizing complexity class with bounded time or space, including \stateBQP{} and $\stateUPSPACE$:

\begin{proposition}
    \label{prop:stateBQP-perfect-completeness}
    For any $0 \leq \delta(n),\gamma(n) \leq 1$, choosing $\delta'\coloneqq \gamma \delta+1-\gamma$, then 
    \[\stateBQP_{\delta}[\gamma]\subseteq \stateBQP_{\delta'}[1] \text{ and } \stateUPSPACE_{\delta}[\gamma] \subseteq \stateUPSPACE_{\delta'}[1].\] 
\end{proposition}

Moreover, we note that achieving perfect completeness has little effect on the distance parameter when combined with the error reduction techniques employed by the aforementioned classes. Specifically, the new distance parameter $\delta'$ can be made \textit{exponentially close} to $\delta$. 

\begin{proof}[Proof of \Cref{prop:stateBQP-perfect-completeness}]
Let $\rho$ be the resulting state upon acceptance with respect to either a \stateBQP{} circuit or a $\stateUPSPACE$ circuit, denoted as $Q_x$, and similarly let $\rho_*$ be the resulting state upon rejection. 
Also, let $\hat{\rho}$ be the resulting state when we replace the output qubit with a qubit in state $\ket{1}$. 
Let $\gamma'$ be the exact success probability of $Q_x$. Recall that $\ket{\psi_x}$ is the target state that $Q_x$ aims to approximate, then we obtain: 
\begin{align*}
\td(\ket{\psi_x}\bra{\psi_x},\hat{\rho}) &= 
\td(\ket{\psi_x}\bra{\psi_x},\gamma'\cdot\rho+(1-\gamma')\cdot \rho_*)\\
&\leq \gamma'\cdot \td(\ket{\psi_x}\bra{\psi_x},\rho) + (1-\gamma')\cdot \td(\ket{\psi_x}\bra{\psi_x},\rho_*)\\
&= \gamma\cdot \td(\ket{\psi_x}\bra{\psi_x},\rho)+ (\gamma'-\gamma)\cdot \td(\ket{\psi_x}\bra{\psi_x},\rho)  + (1-\gamma')\cdot \td(\ket{\psi_x}\bra{\psi_x},\rho_*)\\
&\leq \gamma\delta+1-\gamma. 
\end{align*}
Here, the second line follows from the convexity of the trace distance, and the last line holds because $\td(\psi_x,\rho) \leq \delta(n)$, with the trace distance bounded by $1$. 
\end{proof}

\subsection{Quantum Merlin-Arthur proof systems for synthesizing quantum states}
\label{subsec:stateQMA-and-more}

We will now discuss \stateQMA{}, which is a state-synthesizing counterpart of the class \NP{}. Moreover, it is a natural subclass of \stateQIP{}, as defined~\cite{RY22}, since we merely admit one message from the prover to the verifier. It is worth noting that quantum state-synthesizing complexity classes can be viewed as quantum analogs of classical function classes, such as \textsf{TFNP},\footnote{The class \textsf{TFNP} consists of function problems of the following form: given an input $x$ and a polynomial predicate $F(x,y)$, output any $y$ such that $F(x,y)$ holds, with the promise that such a $y$ exists.} which do not contain \textit{no} instances.

\begin{definition}[$\stateQMA_{\delta}{[c,s]}$]
    \label{def:stateQMA}
    Consider efficiently computable functions $c(n),s(n),\delta(n)$ satisfying $c(n)-s(n)\geq 1/\poly(n)$, $0 \leq s(n) < c(n) \leq 1$, and $0 \leq \delta(n) \leq 1$. 
    A family of quantum states $\{\ket{\psi_x}\}_{x \in \calL}$ is in $\stateQMA_{\delta}[c,s]$ if each $\ket{\psi_x}$ is an $n$-qubit state where $n\coloneqq |x|$, and there is a uniformly generated polynomial-size unitary quantum circuits family $\{V_x\}_{x \in \calL}$ acting on $m+k$ qubits, where $m$ is the number of workspace qubits and $k$ is the number of ancillary qubits where both $m(n)$ and $k(n)$ are polynomials of $n$, such that the following conditions hold: 
    \begin{itemize}
    \item \textbf{\emph{Completeness.}} There is an $m$-qubit state $\ket{w}$ such that 
    $\Pr{V_x \text{ accepts } \ket{w}} \geq c(n)$. 
    \item \textbf{\emph{Soundness.}} For any $m$-qubit state $\ket{w}$ such that $\td(\rho_{x,w}, \ket{\psi_x}\bra{\psi_x}) \geq \delta(n)$, we have $\Pr{V_x \text{ accepts } \ket{w}} \leq s(n)$, where $\rho_{x,w}$ is the resulting state of $V_x$ on the input $\ket{w}\ket{0^k}$ conditioned on the measurement outcome of the output qubit being $1$.  Here, we allow the verifier to trace out some qubits in the resulting state for convenience. 
\end{itemize}
In addition, we use the notation $\stateQMA_{\delta}[l,c,s]$ to represent \stateQMA{} with witness states of bounded length, where $l(n)$ is the number of qubits employed by the witness states. If $l(n)$ is a polynomial function of $n\coloneqq |x|$, we will omit the parameter $l$.
\end{definition}

We define $\stateQMA_{\delta}\coloneqq \stateQMA_{\delta}{[2/3,1/3]}$. Additionally, we also define the variant of \stateQMA{} with exponentially small promise gap\footnote{To be specific, the gap between the acceptance probability in completeness and soundness is inverse-exponential.} and the variant of \stateQMA{} with a logarithmic-length witness in \Cref{def:statePreciseQMA}. 

\begin{definition}[\stateQMA{} with tiny promise gap and short-length witness]
    \label{def:statePreciseQMA}
    \begin{align*}
    \statePreciseQMA_{\delta} &\coloneqq  \cup_{c-s\geq \exp(-\poly(n))} \stateQMA_{\delta}[c,s].\\ 
    \stateQMA_{\delta}[\log,c,s]&\coloneqq  \cup_{l(n) \leq O(\log{n})} \stateQMA_{\delta}[l,c,s].
    \end{align*}
\end{definition}

Furthermore, we need to define $\stateQMAPSPACEoff$, a polynomial-space-bounded variant of \stateQMA{}, which corresponds to state families that are preparable by a polynomial-space-uniform unitary quantum circuit family $\{V_x\}_{x\in\calL}$. In particular, we utilize circuit families outlined in \Cref{def:statePSPACE} as verification circuits, where an $m(n)$-qubit witness state and $k(n)$ ancillary qubits (instead of the all-zero state) serve as input, with the witness length $m(n)$ polynomial in $n$. It is noteworthy that $\stateQMAPSPACEoff$ is not a state-synthesizing counterpart of \NPSPACE{}, as explained in \Cref{remark:online-vs-offline}. 

\begin{remark}[Online vs.\ offline access to quantum witness state]
    \label{remark:online-vs-offline}
    Online and offline access to a classical witness in the class \NP{} are equivalent, but the models significantly differs in space-bounded classical computation and so do quantum counterparts.\footnote{See Section 5.3.1 in~\cite{Goldreich2008} for elaborations on space-bounded classical computation. Regarding quantum scenarios, recent advancements in~\cite{GR23} signifies that quantum analogues of \NPSPACE{} with online access to an exponential-length witness are more powerful than the offline-access variant $\mathsf{QMAPSPACE}$ defined in \cite{FR21}, implying a quantum analogue of Savitch's theorem is unlikely to hold.} The class \NPSPACE{} has online access to an \textit{exponential-length} classical witness, whereas $\stateQMAPSPACEoff$ has offline access to a \textit{polynomial-length} quantum witness, suggesting $\stateQMAPSPACEoff$ is not a state-synthesizing analogue of \NPSPACE{}. 
\end{remark}

By restricting witnesses for both completeness and soundness conditions in \Cref{def:stateQMA} to be classical (i.e., binary strings), we obtain the class \stateQCMA{}, and similarly for the precise variant \statePreciseQCMA{} that has an exponentially small promise gap. 
Likewise, we define $\stateQCMA_{\delta}\coloneqq \stateQCMA_{\delta}[2/3,1/3].$ 

\paragraph{Relation between \stateBQP{} and \stateQMA{}.}
Finally, we remark that while $\stateBQP_{\delta}[\gamma]$ is trivially contained in $\stateQMA_{\delta}[\gamma,\gamma]$ by letting the verifier simply trace out the witness state, proving inclusion in $\stateQMA_{\delta}[\gamma,\gamma']$ with $\gamma(n)-\gamma'(n) \geq 1/\poly(n)$ is more complicated.
In particular, we measure the input quantum state of $Q_x$ (also viewed as a quantum witness state for a \stateQMA{} verifier) on the computational basis, and reject if the outcome is not an all-zero string. We now proceed with a formal statement with a proof:

\begin{proposition}
    \label{prop:stateBQP-in-stateQMA}
    For any $1/\poly(n) \leq \gamma(n) \leq 1$ and $0 \leq \delta \leq 1$,
    $\stateBQP_{\delta}[\gamma] \subseteq \stateQMA_{\delta}[\gamma,\gamma']$ for some $\gamma'(n) > 0$ such that $\gamma(n) - \gamma'(n) \geq 1/\poly(n)$.
\end{proposition}

\begin{proof}
For any given $\stateBQP_{\delta}[\gamma]$ circuit $C_x$ with $m(n)$ workspace qubits, we construct a new $\stateQMA{}$ verification circuit $V_x$. We first introduce a new register $\sfW$ for the $m(n)$-qubit witness state $\ket{w}$, and the remaining is the given \stateBQP{} circuit. Then the construction of $V_x$ follows from \Cref{algo:stateBQP-in-stateQMA}. 

	\begin{algorithm}[ht!]
        1. Perform the \stateBQP{} circuit $C_x$ without the final measurement\; 
        2. Measure all qubits in $\sfW$ (the witness state $\ket{w}$) on the computational basis. \\ ~~~~Reject if the measurement outcome is not $0^m$\;
        3. Measure the designated output qubit in the \stateBQP{} circuit. \\ ~~~~Accept and produce the resulting state if the measurement outcome is $1$, otherwise reject. 
        \BlankLine
		\caption{\stateQMA{} verification circuit $V_x$}
		\label[algorithm]{algo:stateBQP-in-stateQMA}
	\end{algorithm}

It suffices to analyze this protocol. For the completeness condition, it is evident that the optimal witness is $0^m$. Then guaranteed by the given \stateBQP{} circuit, we obtain $\Pr{V_x \text{ accepts } 0^m}\geq\gamma$. For the soundness condition, since any witness $\ket{w}$ that is orthogonal to $\ket{0^m}$ will be simply rejected, then we have derived that 
    \[\Pr{V_x \text{ accepts } \ket{w}}
    =|\innerprod{w}{0^m}|^2\cdot \Pr{C_x \text{ accepts}} \geq \gamma |\innerprod{w}{0^m}|^2.\]
Therefore, for any witness state $\ket{w}$, we conclude that $\Pr{V_x \text{ accepts } \ket{w}} > 0$, and the resulting state $\rho_{x,w}$ satisfies $\td(\rho_{x,w},\ket{\psi_x}\bra{\psi_x})=\td(\rho_{x,0^m},\ket{\psi_x}\bra{\psi_x})\leq \delta$ as desired. 
\end{proof}

\section{Quantum singular value discrimination, revisited}

In this section, we aim to prove \Cref{thm:meta-singular-value-discrimination}, which serves as a ``meta theorem'' for singular value discrimination using different approximation polynomials with varying parameters. 
Our approach differs from Theorem 3.2.9 in~\cite{Gilyen19} as we employ the \textit{projected unitary encoding} with an \textit{odd polynomial} to preserve the output of the verification circuit, as specified in \Cref{def:stateQMA}. 

\vspace{1em}
We will start by defining a \textit{projected unitary encoding}. We say that $U$ is an $\varepsilon$-projected unitary encoding of linear operator $A$ if $\|A- \tilde{\Pi} U \Pi\|_{\infty} \leq \varepsilon$ where orthogonal projectors $\tilde{\Pi}$ and  $\Pi$ may also act on some ancillary qubits. In particular, we refer to $U$ as an exact projected unitary encoding of $A$ when $\varepsilon=0$. 

\begin{restatable}[Singular value discrimination with bounded time and space]{theorem}{singularValueDiscrimination}
    \label{thm:meta-singular-value-discrimination}
    Consider $0 \leq a < b \leq 1$ and an exact projected unitary encoding $A \coloneqq  \tilde{\Pi}U\Pi$, where $U$ acts on $s_1(n)$ qubits. Let $\ket{\psi}$ be an unknown quantum state, a right singular vector of $A$ with a singular value either below $a$ or above $b$.
    Suppose there exists a degree-$d$ odd polynomial $S$ satisfying $|S(x) - \sign(x)| \leq \epsilon$ for all $x \in [-1,1] \setminus (-\delta,\delta)$, where $d = O\left(\delta^{-1}\log(1/\epsilon)\right)$ and $\delta$ linearly depends on $b-a$. 
    We can implement the QSVT associated with the polynomial $S$ to distinguish between the two cases, ensuring that the singular value of the given right singular vector $\ket{\psi}$ is either below $a$ or above $b$, with error probability at most $\epsilon$. 
    
    \noindent In particular, the description of quantum circuit implementation (classical prepossessing) can be computed in deterministic time $\poly(d)$ and space $s_2(n)$, and the circuit implementation takes $\poly(d)$ quantum gates and acts on at most $s_1(n)+O(\log{d})+1$ qubits, where $s_2(n)$ is decided by the construction of $S$.
\end{restatable}

Our construction is heavily influenced by~\cite{GSLW19} (and \Cref{lemma:space-efficient-approxSign} taken from~\cite{MY23}) and we utilize specific statements from~\cite{Gilyen19} for ease of use. We will show that, for any $\epsilon$ equal to $1/\poly(n)$ or $1/\exp(-\poly(n))$, when equipped with different approximation polynomials for the sign function, the space complexity parameters $s_1(n)$ and $s_2(n)$ remain $\poly(n)$ whereas the time complexity parameter $t(n)$ may differ greatly. 

\subsection{A general framework}
\label{sec:singular-value-discrimination}

Let us now define the singular value decomposition for projected unitary encodings. 

\begin{lemma}[Singular value decomposition of a projected unitary, adapted from Definition 2.3.1 in \cite{Gilyen19}]
    \label{def:SVD-projected-unitary}
    Consider a projected unitary encoding of $A$, denoted by $U$, associated with orthogonal projectors $\Pi$ and $\tilde{\Pi}$ on a finite-dimensional Hilbert space $\calH_U$. 
    Namely, $A=\tilde{\Pi} U \Pi$. 
    Then there exist orthonormal bases of $\Pi_i$ and $\tilde{\Pi}_i$ as follows: for $\Pi_i$, the orthonormal basis of the subspace $\Img(\Pi_i)$ is $\left\{ \ket{\psi_{i}}: i \in [d] \right\}$, where $d\coloneqq \rank(\Pi)$. Similarly, for $\tilde{\Pi}_i$, the orthonormal basis of the subspace $\Img(\tilde{\Pi}_i)$ is $\big\{ \ket{\tilde{\psi}_i}: i\in[\tilde{d}] \big\}$, where $\tilde{d}\coloneqq \rank(\tilde{\Pi})$.
    These bases ensure that there is a singular value decomposition $A = \sum_{i=1}^{\min\{d,\tilde{d}\}} \sigma_i \ket{\tilde{\psi}_i}\bra{\psi_i}$, where singular values $\sigma_i > \sigma_j$ for any $i < j \in [\min\{d,\tilde{d}\}]$. 
\end{lemma}

With these definitions in place, we present the alternating phase modulation as \Cref{lemma:phase-modulation}, which serves as the key ingredient for implementing quantum singular value transformations. 
It is worth noting that Lemma \ref{lemma:phase-modulation} deviates from Theorem 2.3.7 in \cite{Gilyen19} due to our assumption that the sequence of rotation angles is already provided. The rotation angles $\Phi$ in this context correspond to the polynomial $P$ and will be utilized in computing the descriptions of quantum circuits implementing QSVT corresponding to $P$.

\begin{lemma}[QSVT by alternating phase modulation, adapted from Theorem 2.3.7 in~\cite{Gilyen19}]
    \label{lemma:phase-modulation}
    Let $P\in\bbC[x]$ be an odd polynomial, and let $\Phi\in\bbR^n$ be the corresponding sequence of rotation angles. There exists an explicit implementation of $P^{\rm (SV)}(\tilde{\Pi} U \Pi) = \tilde{\Pi} U_{\Phi} \Pi$ when $n$ is odd, using a single ancillary qubit. Specifically, $U_{\Phi}$ represents a quantum circuit that implements the QSVT associated with $P$, using the corresponding rotation angles specified by $\Phi$. Moreover, the notation $P^{\rm (SV)}(A)$ indicates that the polynomial $P$ is applied to the singular values of the matrix $A$, rather than to its eigenvalues.
\end{lemma}

We will now proceed with proving \Cref{thm:meta-singular-value-discrimination}. 

\begin{proof}[Proof of \Cref{thm:meta-singular-value-discrimination}]

    Given an exact projected unitary encoding $\tilde{\Pi}U\Pi$, with a singular value decomposition $W\Sigma V^\dagger=\tilde{\Pi}U\Pi$, and using with \Cref{lemma:phase-modulation}, it suffices to construct an odd polynomial $P$ associated with a sequence of angles $\Phi\in\bbR^m$, where $m=O(\delta^{-1}\log(1/\varepsilon))$. The goal is to ensure that the following conditions hold: 
    \begin{align*}
        \big\|\tilde{\Pi}_{\geq t+\delta} U_{\Phi} \Pi_{\geq t+\delta} - I\otimes \sum_{i:\sigma_i \geq t+\delta} \ket{\tilde{\psi}_i}\bra{\psi_i} \big\| \leq \varepsilon \text{ and }
        \|\big( \bra{+}\otimes\tilde{\Pi}_{\leq t-\delta}\big) U_{\Phi} \left( \ket{+}\otimes \Pi_{\leq t-\delta} \right) - 0\| \leq \varepsilon. 
    \end{align*}
    
    Next, we define singular value threshold projectors as $\Pi_{\geq \delta}\coloneqq \Pi V \Sigma_{\geq \delta} V^{\dagger} \Pi$, and similarly $\Pi_{\leq\delta}\coloneqq \Pi V \Sigma_{\leq \delta} V^{\dagger} \Pi$. In the same way, we define $\tilde{\Pi}_{\geq\delta}\coloneqq \Pi U \Sigma_{\geq \delta} U^{\dagger} \Pi$ and  $\tilde{\Pi}_{\leq\delta} \coloneqq \Pi U \Sigma_{\leq \delta} U^{\dagger} \Pi$. 
    Using these constructions, we apply an $\epsilon$-approximate singular value projector by choosing $t=(a+b)/2$ and $\delta=(b-a)/2$. Then, we measure $\ketbra{+}{+}\otimes \Pi$: 
    \begin{itemize}
        \item If the final state lies in $\Img(\ketbra{+}{+}\otimes \Pi)$, there exists a singular value $\sigma_i$ above $b$; 
        \item Otherwise, all singular values $\sigma_i$ must be below $a$. 
    \end{itemize}
    
    The remaining task is to implement singular value threshold projectors. It is noteworthy that $U_{\Phi}$ can be implemented using a single ancillary qubit, which requires $m$ uses of $U$,$U^\dagger$, $\rm C_\Pi NOT$, $\rm C_{\tilde{\Pi}} NOT$, and single qubit gates. 

    \paragraph{Implementing singular value threshold projectors.} We begin by constructing an odd polynomial $P\in\bbR[x]$ of degree $m=O(\log(1/\epsilon^2)/\delta)$, which approximates the function 
    \[\frac{1}{2}[(1-\epsilon)\cdot\sign(x+t) + (1-\epsilon)\cdot\sign(x-t) + 2\epsilon \cdot\sign(x)]\] 
    with $\epsilon^2/4$ precision on the interval $[-1,1] \setminus (-t-\delta,-t+\delta) \cup (t-\delta,t+\delta)$. 

    This construction of $P$ is based on the degree-$d$ approximation polynomial $S(x)$, which is specified in the statement of the theorem. This polynomial $S(x)$ satisfies the condition: 
    \[\forall x \in [-1,1] \setminus (-\delta,\delta), |S(x)-\sign(x)| \leq \epsilon, \text{ where } d=O(\delta^{-1} \log(1/\epsilon)).\] 
    In addition, it holds that $|P(x)| \leq 1$ for any $-1 \leq x \leq 1$, with the following properties: 
    \begin{align*}
        (-1)^z P(x) \in [0,\epsilon] &\text{ when } (-1)^z x \in [0,t-\delta],\\
        (-1)^z P(x) \in [1-\epsilon,1] &\text{ when } (-1)^z \in [t+\delta,1], \text{ for } z \in \binset. 
    \end{align*}
    
    Finally, it remains to construct a projected unitary encoding $U$ such that 
    \[\|\tilde{\Pi} U_{\Phi}\Pi - P^{\rm (SV)}(\tilde{\Pi} U \Pi)\| \leq \epsilon.\] This task can be accomplished using \Cref{lemma:phase-modulation}, provided that the rotation angles $\Phi$, corresponding to $P(x)$, can be computed in deterministically $\poly(d)$ time and $s_2(n)$ space. 
    Furthermore, the gate complexity follows from Lemma 2.3.9 in~\cite{Gilyen19}, which provides an explicit implementation and analysis of the gate complexity for the alternating phase modulation procedure described in \Cref{lemma:phase-modulation}. 
\end{proof}

\subsection{Scenarios in time-efficient and space-bounded}

Equipped with \Cref{thm:meta-singular-value-discrimination}, utilizing a polynomial approximation of the sign function in Corollary 6 of~\cite{LC17}, as well as the angle finding algorithms in Theorems 2.2.1-2.2.3 of ~\cite{Gilyen19}, we conclude the time-efficient singular value discrimination. 

\begin{theorem}[Time-efficient singular value discrimination]
    \label{thm:time-efficient-singular-value-discrimination}
    Consider $0\leq a,b \leq 1$ and an exact projected unitary encoding $A\coloneqq \tilde{\Pi} U \Pi$. Let $\ket{\psi}$ be an unknown state, a right singular vector of $A$ with a singular value either below $a$ or above $b$ such that $b-a \geq 1/\poly(n)$. Employing the quantum singular value transform (QSVT) with a degree-$O\big(\frac{\log{1/\epsilon}}{b-a}\big)$ odd polynomial, one can distinguish the two cases with error probability at most $\epsilon$. Moreover, the time complexity of implementing this QSVT is $\poly\log(1/\epsilon)\cdot\poly(1/(b-a))$. We can also compute the description of this quantum circuit implementation in deterministic time $\poly\log(1/\epsilon)\cdot\poly(1/(b-a))$.     
\end{theorem}

\begin{proof}[Proof Sketch]
It suffices to construct a degree-$d$ approximation polynomial of the sign function and to find the rotation angles used in \Cref{lemma:phase-modulation} in deterministic time $\poly\log(1/\epsilon)\cdot\poly(1/(b-a))$. Lastly, the construction of the approximation polynomial is provided in~\cite{GSLW19,Gilyen19}, along with \Cref{prop:poly-approx-sign}. 

\begin{proposition}[Polynomial approximation of the sign function, Corollary 6 in \cite{LC17}]
    \label{prop:poly-approx-sign}
    For all $\delta > 0$, $\epsilon \in (0,1/2)$, there exists an efficiently computable odd polynomials $P \in \bbR[x]$ of degree $n=O(\log(1/\epsilon)/\delta)$ s.t.  
    $\forall x \in [-2,2], |P(x)| \leq 1$ and $\forall x \in[-2,2] \setminus (-\delta, \delta): |P(x)-\sign(x)|\leq \epsilon$.
\end{proposition}

Regarding finding angles, this is achievable in time $\tilde{O}(d^3\poly\log(1/\epsilon))$ using the recent developments~\cite{Haah19,CDG+20}. This results in a quantum circuit implementing the desired QSVT, and the description of this quantum circuit implementation can be computed in \Ptime{}.  
\end{proof}

Now we move to \textit{the space-bounded scenario}, stated in \Cref{thm:space-bounded-singular-value-discrimination}. 

\begin{theorem}[Space-bounded singular value discrimination]
    \label{thm:space-bounded-singular-value-discrimination}
    Consider $0\leq a,b \leq 1$ and an exact projected unitary encoding $A\coloneqq \tilde{\Pi} U \Pi$. Let $\ket{\psi}$ be an unknown state, a right singular vector of $A$ with a singular value either below $a$ or above $b$ such that $b-a \geq \exp(-\poly(n))$. Employing the quantum singular value transform (QSVT) with a degree-$O\big(\frac{\log{1/\epsilon}}{b-a}\big)$ odd polynomial, one can distinguish the two cases with error probability at most $\epsilon$. 
    In addition, quantum circuits implementing this QSVT utilize $\poly(n)$ qubits, and we can compute the description of this quantum circuit implementation in \PSPACE{}.  
\end{theorem}

The proof of \Cref{thm:space-bounded-singular-value-discrimination} is closely derived from the \textit{space-bounded} quantum singular value transformation techniques in~\cite[Lemma 3.13]{MY23}. This is because their proof techniques can be straightforwardly adapted to the context of projected unitary encodings. In addition, we require the following lemma to proceed with the proof:

\begin{lemma}[Exponentially good approximation to the sign function, adapted from Lemma 2.10 in~\cite{MY23}]
    \label{lemma:space-efficient-approxSign}
    For any $\epsilon \geq 2^{-\poly(n)}$, there exists a degree $d=O(\epsilon^{-1}\log(1/\epsilon))=O(2^{\poly(n)})$ and \PSPACE{}-computable coefficients $c_0,\cdots,c_d$ such that 
    $\forall x \in [-1,1] \setminus [-\epsilon,\epsilon]$, $\left|\sign(x)-P^{\sign}_d(x)\right| \leq \epsilon$  where  $P_d^{\sign}\coloneqq \sum_{i=0}^d c_i T_i$ and $T_i(x)$ is the Chebyshev polynomial (of the first kind).\footnote{The Chebyshev polynomials (of the first kind) $T_k(x)$ are defined via the following recurrence relation: $T_0(x)=1$, $T_1(x)=x$, and $T_{k+1}(x)=2x T_k(x)-T_{k-1}(x)$. For $x \in [-1,1]$, an equivalent definition is $T_k(\cos \theta) = \cos(k \theta)$.}   
    Furthermore, the coefficient vector $c=(c_1,\cdots,c_d)$ has norm bounded by $\|c\|_1 \leq O(\log d)$.
\end{lemma}

We now outline the proof sketch of \Cref{thm:space-bounded-singular-value-discrimination}: 

\begin{proof}[Proof Sketch of \Cref{thm:space-bounded-singular-value-discrimination}]
Similar to the time-efficient scenario (e.g., \Cref{thm:time-efficient-singular-value-discrimination}), it suffices to implement the QSVT corresponding to the sign function, provided that we have an exponentially good (bounded) polynomial approximation of the sign function, as stated in \Cref{lemma:space-efficient-approxSign}. 

However, instead of directly constructing the sequence of rotation angles in the time-efficient scenario, we employ the linear combination of unitaries (LCU) technique~\cite{BCCKS14}. This technique can be readily adapted to space-bounded scenarios (as demonstrated in Lemma 3.6 in~\cite{MY23}) and is also applicable to projected unitary encodings. Next, we focus on implementing the quantum singular value transformation (QSVT) for the Chebyshev polynomials $T_i(x)$ with odd $i\in [1,d]$ utilized in~\Cref{lemma:space-efficient-approxSign}. To achieve this, we examine the proof of Lemma 2.2.7 in~\cite{Gilyen19}, which specifies the rotation angles corresponding to Chebyshev polynomials, and directly adapt it to the context of projected unitary encodings. Consequently, implementing an exponential-degree QSVT corresponding to the sign function requires polynomially many qubits, and the description of this circuit implementation can be computed in \PSPACE{}.
\end{proof}

\section{Doubly-preserving error reduction for state-synthesizing classes}
\label{sec:error-reduction-for-stateQMA}

In this section, we will present error reduction for \stateQMA{} and its variants which preserves not only \textit{the (quantum) witness state} but also \textit{the resulting state} that well-approximates the target state.  

\subsection{Doubly-preserving error reduction for \stateQMA{} and more}

We start by stating error reduction for \stateQMA{}, this leads to 
\[\stateQMA{}_\delta = \cup_{c-s\geq 1/\poly(n)} \stateQMA_{\delta}[c,s].\] 

\begin{theorem}[Error reduction for \stateQMA{}]
    \label{thm:error-reduction-stateQMA}
    For any efficiently computable $c(n),s(n),\delta(n)$ such that $0 \leq s(n) < c(n) \leq 1$, $c(n)-s(n)\geq 1/\poly(n)$, we have that for any polynomial $l(n)$, 
    \[\stateQMA_{\delta}[c,s] \subseteq \stateQMA_{\delta}[1-2^{-l},2^{-l}].\] 
\end{theorem}

It is noteworthy that the proof of \Cref{thm:error-reduction-stateQMA} can be understood as a sequential repetition of the original \stateQMA{} verifier in a clever manner. In particular, the number of repetitions is $O(l/(\sqrt{c}-\sqrt{s}))$, implying that the circuit size of the resultant \stateQMA{} verifier remains polynomial as long as promise gap $c-s$ is at least polynomially small. 

\begin{proof}[Proof of \Cref{thm:error-reduction-stateQMA}] Our proof closely relates to Theorem 38 in~\cite{GSLW19} besides an additional analysis on resulting states. 

\paragraph{Amplifying the promise gap by QSVT.}
Note that the acceptance probability of a \stateQMA{} verifier $V_x$ taking $\ket{w}$ as a quantum witness is 
$\Pr{V_x \text{ accepts } \ket{\psi}} = \| \ket{1}\bra{1}_{\Out} V_x \ket{0^k}\ket{w} \|_2^2 \geq c$ or $\leq s$.
Then consider a projected unitary encoding $M_x\coloneqq \Pi_{\Out} V_x \Pi_{\In}$ such that
$\| M_x \| \geq \sqrt{c}$ or $\sqrt{s}$ where $\Pi_{\In}\coloneqq  (\ketbra{0}{0}^{\otimes k}\otimes I_m)$ and $\Pi_{\Out}\coloneqq  \ket{1}\bra{1}_{\Out} \otimes I_{m+k-1}$. 
Since $\|M_x\|=\sigma_{\max}(M_x)$ where $\sigma_{\max}(M_x)$ is the largest singular value of $M_x$, it suffices to distinguish the cases where the largest singular value of $M_x$ is either below $\sqrt{s}$ or above $\sqrt{c}$. By setting $a=\sqrt{s},b=\sqrt{c}$, and $\varepsilon=2^{-l(n)}$, this task is a straightforward corollary of the time-efficient singular value discrimination (\Cref{thm:time-efficient-singular-value-discrimination}). 

\paragraph{QSVT preserves both the witness state and the resulting state.} 
Utilizing notations of \Cref{def:SVD-projected-unitary}, we notice that the construction in the proof of \Cref{thm:time-efficient-singular-value-discrimination} essentially maps $M_x=\sum_i \sigma_i \ket{\tilde{\psi}_i}\bra{\psi_i}$ to $f(M_x)=\sum_i f(\sigma_i) \ket{\tilde{\psi}_i}\bra{\psi_i}$, for an odd polynomial $f$, such that $f(x)\in[1-\varepsilon,1]$ if $x \geq b$ and $f(x)\in[0,\varepsilon]$ if $x \leq a$ for any $0 \leq x \leq 1$. 
Since both left and right singular vectors are invariant in $f(M_x)$, the resulting state and the witness state are clearly preserved after reducing errors. 
\end{proof}

Using the space-bounded quantum singular value discrimination (\Cref{thm:space-bounded-singular-value-discrimination}), we will now improve \Cref{thm:error-reduction-stateQMA} to $\stateQMAPSPACEoff$, which may have \textit{exponential precision}. This is partially analogous to~\cite{FKLMN16} since \Cref{thm:space-bounded-singular-value-discrimination} works merely for \textit{polynomial space}.\footnote{To achieve a state-synthesizing analog of error reduction for unitary quantum logspace~\cite{FKLMN16}, a \textit{space-efficient} version of \Cref{thm:space-bounded-singular-value-discrimination} is necessary. Such a space-efficient variant of QSVT~\cite{LGLW23} became available only after our work. By replacing \Cref{thm:space-bounded-singular-value-discrimination} with this technique in the proof of \Cref{thm:error-reduction-bounded-stateQMA} and \Cref{thm:error-redution-stateUPSPACE}, we can establish error reduction for state-synthesizing unitary quantum \textit{logspace} classes $\mathsf{stateBQ_{U}L}_{\delta}$ and $\mathsf{stateQMA_{U}L}_{\delta}^\mathsf{off}$, as well as the equivalence $\mathsf{stateQMA_{U}L}_{\delta}^\mathsf{off}=\mathsf{stateBQ_{U}L}_{\delta}$, which is fully analogous to~\cite{FKLMN16}. }

\begin{theorem}[Error reduction for $\stateQMAPSPACEoff$]
    \label{thm:error-reduction-bounded-stateQMA}
    Given any efficiently computable functions $c(n)$ and $s(n)$, also $\delta(n)$ such that $0 \leq s(n) < c(n) \leq 1$ and $c(n)-s(n) \geq \exp(-\poly(n))$, then for any polynomial $l(n)$,
    \[\stateQMAPSPACEoff_{\delta}[c,s] \subseteq \stateQMAPSPACEoff_{\delta}[1-2^{-l},2^{-l}].\] 
\end{theorem}

In addition, by forcing the input state of the ``verification circuit'' to be an all-zero state in the proof of \Cref{thm:error-reduction-stateQMA}, namely replacing the projector $\Pi_{\In}\coloneqq  (\ketbra{0}{0}^{\otimes k}\otimes I_m)$  with $\Pi'_{\In}\coloneqq  \ketbra{0}{0}^{\otimes k+m}$, we will straightforwardly result in error reduction for \stateBQP{}. 

\begin{theorem}[Error reduction for \stateBQP{}]
\label{theorem:error-reduction-stateBQP}
For any polynomials $p(n)$ such that $1/p(n) \leq \gamma(n) < 1$, we have that for any polynomial $l(n)$, 
\[\stateBQP_{\delta}[\gamma] \subseteq \stateBQP_{\delta}[1-2^{-l}].\] 
\end{theorem}

\Cref{theorem:error-reduction-stateBQP} implies that $\stateBQP_{\delta} = \cup_{1 \leq \gamma^{-1} \leq \poly(n)} \stateBQP_{\delta}[\gamma]$. Also, combining with \Cref{prop:stateBQP-perfect-completeness}, \stateBQP{} achieves perfect completeness with a slightly modified distance parameter $\delta'$. Particularly, $\stateBQP_{\delta}[\gamma] \subseteq \stateBQP_{\delta'}[1]$ where $\delta'(n)\coloneqq \delta(n)+\exp(-\poly(n))$. 

Likewise, together with the projector $\Pi'_{\In}\coloneqq  \ketbra{0}{0}^{\otimes k+m}$ and the space-bounded quantum singular value discrimination (\Cref{thm:space-bounded-singular-value-discrimination}), we will deduce error reduction for $\stateUPSPACE$.  

\begin{theorem}[Error reduction for $\stateUPSPACE$]
    \label{thm:error-redution-stateUPSPACE}
    For any efficiently computable functions $\gamma(n)$ and $\delta(n)$ such that $\exp(-\poly(n)) \leq \gamma(n) \leq 1$, we have that for any polynomial $l(n)$,
    \[\stateUPSPACE_{\delta}[\gamma] \subseteq \stateUPSPACE_{\delta}[1-2^{-l}].\] 
\end{theorem}

A direct consequence of \Cref{thm:error-redution-stateUPSPACE} is that $\stateUPSPACE_{\delta} = \cup_{\gamma \geq 2^{-\poly(n)}} \stateUPSPACE_{\delta}[\gamma]$. Similarly, by utilizing \Cref{prop:stateBQP-perfect-completeness}, we can also conclude that 
\[\stateUPSPACE_{\delta}[\gamma] \subseteq \stateUPSPACE_{\delta'}[1] \text{ where } \delta'(n)=\delta(n)+\exp(-\poly(n)).\]

\subsection{Application 1: \stateQMA{} with a short message is as weak as \stateBQP{}}

Recall that $\stateQMA_{\log}$ is a variant of \stateQMA{} where the witness state is logarithmic-size. 
Analogous to Theorem 3.10 in~\cite{MW05}, a short message is also useless for \stateQMA{}. 
\begin{theorem}[A short message is useless for \stateQMA{}]
\label{thm:stateQMAlog-in-stateBQP}
For any $0 \leq s(n) < c(n) \leq 1$ and $c(n)-s(n) \geq 1/\poly(n)$, there exists a polynomial $q(n)$ such that
\[\stateQMA_{\delta}[\log,c,s] \subseteq \stateBQP_{\delta}[1/q].\] 
\end{theorem}
\begin{proof}
    Consider a state family $\{\ket{\psi_x}\}_{x\in\calL}$ that is in $\stateQMA_{\delta}[\log,c,s]$, we then notice that this state family is also in $\stateQMA_{\delta}[\log,1-2^{-p},2^{-p}]$ where $p(n)$ is a polynomial of $n\coloneqq |x|$ after performing error reduction (\Cref{thm:error-reduction-stateQMA}). Now let $\{V_x\}_{x \in \calL}$ be the family of quantum verifiers with negligible errors. 

    \paragraph{Removing the witness by a ``random-guess'' state.} 
    Now consider a \stateBQP{} algorithm that applies $V_x$ with the witness being a completely mixed state $\tilde{I}_m\coloneqq 2^{-m}I_m$ on $m(n)=O(\log{n})$ qubits. It accepts iff $V_x$ accepts. 
    For the analysis, we define $M_x\coloneqq (\ket{1}\bra{1}_{\Out}\otimes \ketbra{0}{0}^{\otimes m+k-1}) V_x (I_m\otimes \ketbra{0}{0}^{\otimes k})$. Then $\Pr{V_x \text{ accepts} \ket{w}}=\|M_x \ket{w}\|_2^2$, which implies that the acceptance probability of the \stateBQP{} algorithm is
    \begin{equation}
    \label{eq:pacc-stateQMAlog-in-stateBQP}
    \mathrm{Pr}\big[V_x \text{ accepts } \tilde{I}_m\big]=\Tr(M_x^{\dagger} M_x 2^{-m}I_m)=2^{-m} \Tr(M_x^{\dagger} M_x)=2^{-m}\lambda_{\max}(M_x^{\dagger}M_x),
    \end{equation}
    where $\lambda_{\max}(M_x^{\dagger}M_x)$ is the largest eigenvalue of $M_x^{\dagger}M_x$. 
    For the completeness, there exists $\ket{w}$ such that $\lambda_{\max}(M_x^{\dagger}M_x) \geq \Pr{V_x \text{ accepts } \ket{w}} \geq 1-2^{-p(n)}$. Plugging it into \Cref{eq:pacc-stateQMAlog-in-stateBQP}, we have $\Pr{V_x \text{ accepts } \tilde{I}_m} \geq 2^{-m(n)}(1-2^{-p(n)}) \geq 1/q(n)$, where $q(n)\coloneqq 2^{O(m(n))}$ is a polynomial of $n\coloneqq |x|$. 
    For the soundness, we obtain $\Pr{V_x \text{ accepts } \tilde{I}_m} \geq 2^{-p(n)}$, it guarantees that $\td(\rho_{\tilde{I}_m},\psi_x) \leq \delta$ owing to \Cref{def:stateQMA}. 
\end{proof}

We remark that \Cref{thm:stateQMAlog-in-stateBQP} straightforwardly adapts to any \stateQMA{} verifier with witness states of polynomial-size, which is a state-synthesizing counterpart of $\QMA{} \subseteq \PP$~\cite{Vyalyi03,MW05}.

\begin{proposition}[``$\stateQMA{} \subseteq \mathsf{statePP}$'']
$\stateQMA{}_\delta \subseteq \statePreciseBQP{}_\delta$. 
\end{proposition}

\begin{proof}
    Consider a family of quantum states $\{\ket{\psi_x}\}_{n\in\bbN} \in \stateQMA_m[c,s,\delta]$ where $m(n)$ is the size of the witness state. Analogous to \Cref{thm:stateQMAlog-in-stateBQP}, we have derived that
    \[\stateQMA_{\delta}[m,c,s] \subseteq \stateQMA_{\delta}[m,1-2^{-m'}, 2^{-m'}] \subseteq \stateBQP_{\delta}[2^{-m}],\]
    where $m'(n)=m(n)\cdot n^2$. 
\end{proof}

\subsection{Application 2: \statePreciseQMA{} is in \statePSPACE{}}

Finally, we provide a state-synthesizing analogue of $\PreciseQMA\subseteq\BQPSPACE$~\cite{FL16,FL18}.

\begin{theorem}
    \label{thm:statePreciseQMA-in-statePSPACE}
    $\statePreciseQMA_{\delta} \subseteq \stateUPSPACE_{\delta}$.
\end{theorem}

\begin{proof}
    Consider a family of quantum states $\{\ket{\psi_x}\}_{x\in\calL}$ that is in $\statePreciseQMA_{\delta}[m,c,s]$ where $m(n)$ is the size of witness state. We begin by observing that 
    \[\statePreciseQMA_{\delta}[m,c,s] \subseteq \mathsf{stateQMA_{U}PSPACE}^\mathsf{off}_{\delta}[m,c,s].\] 
    Then we replace the witness with a ``random-guess'' state analogous to the proof of \Cref{thm:stateQMAlog-in-stateBQP}. Utilizing error reduction for $\stateQMAPSPACEoff$ (\Cref{thm:error-reduction-bounded-stateQMA}), we have derived that 
    \[\statePreciseQMA_{\delta}[m,c,s] \subseteq \mathsf{stateQMA_{U}PSPACE}^\mathsf{off}_{\delta}[m,1-2^{-m'},2^{-m'}] \subseteq \stateUPSPACE_{\delta}[2^{-m}],\]
    where $m'(n)=m(n)\cdot n^2$. Employing error reduction for $\stateUPSPACE$ (\Cref{thm:error-reduction-bounded-stateQMA}), we conclude that 
    \[\statePreciseQMA_{\delta}[m,c,s] \subseteq \stateUPSPACE_{\delta}[2^{-m}] \subseteq \stateUPSPACE_{\delta}[1-2^{-l}],\]
    where $l(n)$ is a polynomial of $n\coloneqq |x|$. This completes the proof. 
\end{proof}


\section{\UQMA{} witness family is in \stateQMA{}}

In this section, we will provide a natural example of \stateQMA{} state families. Specifically, the family of \textit{yes}-instance witness states of verifiers corresponding to $\calL_{\rm yes}$, where $(\calL_{\rm yes},\calL_{\rm no})\in \UQMA$, is in \stateQMA{}: 

\begin{theorem}[$\UQMA_{1-\nu}$ witness family is in \stateQMA{}]
    \label{thm:UQMA-witness-in-stateQMA-scaling}
    For any promise problem $\calL = (\calL_{\rm yes},\calL_{\rm no})$ that is in $\UQMA_{1-\nu}$, the family of unique witness states $\{\ket{w_x}\}_{x\in \calL_{\rm yes}}$ corresponding to \textit{yes} instances is in $\stateQMA_{\delta}[c',s']$, where $c'=(1+\cos{\nu})/2-\epsilon$, $s'=(1+\cos(\lambda_1 t))(2+\cos(\nu t) - \cos(\lambda_1 t))/4+\epsilon$, and $\delta=O(\max\{(1-\cos{\nu})/2,\delta_1,q(\nu\delta_1 t/T^3)\})$ for some polynomial $q$. 
    Here, let $H^{(x)}$ be the Hamiltonian corresponding to a \UQMA{} verifier $V_x$ with the parameter $\delta_1$ specified in \Cref{lemma:Hamiltonian-from-UQMAverifier}, $\lambda_1$ is the second smallest eigenvalue of $H^{(x)}$ respectively, $T$ is the size of verification circuit $V_x$, as well as $t$ satisfies that $\lambda_{\max}\big(H^{(x)}\big) t \leq \pi$, and $\epsilon$ is the implementation error of $\exp(-iH^{(x)}t)$. 
\end{theorem}

By employing \Cref{thm:UQMA-witness-in-stateQMA-scaling}, we can determine that the promise gap of the new \stateQMA{} verifier is given by
\[c'-s' = (1-\cos(\lambda_1 t))(\cos(\nu t)-\cos(\lambda_1 t)) - 2\epsilon \geq q'(\nu\delta_1 t/T^3) - 2\epsilon.\] 
Here, $q'$ represents a polynomial due to the bound $\Delta(H^{(x)}) = \lambda_1-\nu \geq \Omega(\nu\delta_1/T^3)$ stated in \Cref{lemma:Hamiltonian-from-UQMAverifier}. According to \Cref{lemma:space-efficient-Hamiltonian-simulation}, we know that the implementation error $\epsilon$ is exponentially small. Consequently, by utilizing \UQMA{} error reduction (\Cref{thm:UQMA-error-reduction}), we establish $\UQMA \subseteq \UQMA_{1/\poly}$ and thus \Cref{thm:UQMA-witness-in-stateQMA-scaling} implies \Cref{corr:UQMA-witness-in-stateQMA} by choosing $\delta_1 = 1/\poly(n)$: 

\begin{corollary}[$\UQMA$ witness family is in \stateQMA{}]
    \label{corr:UQMA-witness-in-stateQMA}
    For any promise problem $\calL = (\calL_{\rm yes},\calL_{\rm no})$ in $\UQMA$, the family of unique witness states $\{\ket{w_x}\}_{x\in \calL_{\rm yes}}$ corresponding to \textit{yes} instances is in $\stateQMA_{\delta}$, where $\delta(n)=1/\poly(n)$. 
\end{corollary}

Similarly, \Cref{thm:UQMA-witness-in-stateQMA-scaling} implies \Cref{corr:PreciseUQMA-witness-in-statePreciseQMA} by setting $\delta_1=\exp(-\poly(n))$:

\begin{corollary}[$\PreciseUQMA_{1-\negl}$ witness family is in \statePreciseQMA{}]
    \label{corr:PreciseUQMA-witness-in-statePreciseQMA}
    For any promise problem $\calL = (\calL_{\rm yes},\calL_{\rm no})$ in $\PreciseUQMA_{1-\negl}$, the family of unique witness states $\{\ket{w_x}\}_{x\in \calL_{\rm yes}}$ corresponding to \textit{yes} instances is in $\statePreciseQMA_{\delta}$, where $\delta(n)=\exp(-\poly(n))$.  
\end{corollary}

\subsection{Proof of Theorem \ref{thm:UQMA-witness-in-stateQMA-scaling}}

To establish \Cref{thm:UQMA-witness-in-stateQMA-scaling}, we begin with a weighted circuit-to-Hamiltonian construction~\cite{BC18,BCNY19}, as stated in \Cref{lemma:Hamiltonian-from-UQMAverifier}, where the resulting ground state (also known as the history state) provides a good approximation of the output state produced by the \UQMA{} verifier. The proof of \Cref{lemma:Hamiltonian-from-UQMAverifier} can be found in \Cref{subsec:weighted-FK-construction}. 

\begin{lemma}[Weighted circuit-to-Hamiltonian construction]
    \label{lemma:Hamiltonian-from-UQMAverifier}    
    For any $x \in \calL_{\rm yes}$, where $\calL=(\calL_{\rm yes},\calL_{\rm no})$ is a promise problem that is in $\UQMA$, there exists an $O(\log{n})$-local Hamiltonian $H^{(x)}$ with ground energy $\nu$ from the \UQMA{} verifier family $\{V_x\}_{x \in \calL}$ with maximum acceptance probability $1-\nu$, along with the corresponding unique witness state family $\{\ket{w_x}\}_{x\in \calL}$. This unique ground state of $H^{(x)}$ is given by
    \[\ket{\Omega_x} = \sum_{t=0}^{T-1} \frac{\sqrt{\delta_1}}{\sqrt{T}} \ket{t}\otimes U_t\cdots U_1 \ket{w_x} + \sqrt{1-\delta_1} \ket{T} \otimes U_T\cdots U_1 \ket{w_x}.\] 
    Here $T$ denotes the number of quantum gates utilized by $V_x$. 
    Furthermore, the spectral gap of $H^{(x)}$ satisfies $\Delta(H^{(x)}) \geq \Omega(\nu\delta_1/T^{3})$. 
\end{lemma}

With \Cref{lemma:Hamiltonian-from-UQMAverifier} established, we proceed to prove the main theorem in this section.

\begin{proof}[Proof of \Cref{thm:UQMA-witness-in-stateQMA-scaling}]
    We begin by constructing a \stateQMA{} verifier $V'_x$. 
    Let $H^{(x)}$ be the $O(\log n)$-local Hamiltonian, which is obtained by applying \Cref{lemma:Hamiltonian-from-UQMAverifier} to a \UQMA{} verifier $V_x$ with the unique witness $\ket{w_x}$. 
    Our construction is primarily relies on the one-bit precision phase estimation~\cite{Kitaev95}, often referred to as Hadamard test~\cite{AJL09}, as illustrated in \Cref{algo:UQMA-witness-in-stateQMA}. 
    
	\begin{algorithm}[ht!]
        1. Apply one-bit precision phase estimation on registers $\sfO$, $\sfW$, and $\sfC$. Here, register $\sfO$ is the control qubit, while registers $\sfW$ and $\sfC$ are the target qubits. Specifically, (clock) register $\sfC$ consists of $\lceil \log{T} \rceil$ qubits, and the number of qubits in register $\sfW$ matches the number of qubits that $V_x$ acts on\;
        2. Apply the circuit $V_x^{\dagger}$ to register $\sfW$\;
        3. Measure the designated output qubit, and we say that $V'_x$ accepts $\ket{\psi}$ if the final measurement outcome is $0$\; 
        4. If $V'_x$ accepts, trace out all qubits in $\sfC$. The resulting state,  denoted as $\rho_{\psi}$, is located in register $\sfW$
        \BlankLine
		\caption{\stateQMA{} verification circuit $V'_x$}
		\label[algorithm]{algo:UQMA-witness-in-stateQMA}
	\end{algorithm}

    Here, we choose the parameter $t=\frac{\pi}{p(n) \max_{s,t}|H^{(x)}(s,t)|} \leq \frac{\pi}{\|H^{(x)}\|_{\infty}}$ such that all eigenvalues of $H^{(x)} t$ is in $[0,\pi]$. The corresponding circuit implementation has demonstrated in \Cref{fig:stateQMA-new-verifier}.
    
    \begin{figure}[ht!]
    \centering
    \begin{quantikz}[wire types={q,b,q}, classical gap=0.07cm]
	  \lstick{$\ket{0}_{\sfO}$} & \gate{H}  & \ctrl{1} & \gate{H} & \meter{0?} \\
 	  \lstick[wires=2]{$\ket{\psi}_{\sfW,\sfC}$} &  & \gate[2]{\exp(-iH^{(x)} t)} & \gate[][6ex][6ex]{V_x^{\dagger}} & \rstick{$\rho_{\psi}$}\\
        & \qwbundle{\lceil \log{T} \rceil} & & \qwbundle{\lceil \log{T} \rceil} & \trash{\text{trace}}
    \end{quantikz}
    \caption{\stateQMA{} verification circuit $V'_x$}	
    \label{fig:stateQMA-new-verifier}
    \end{figure}
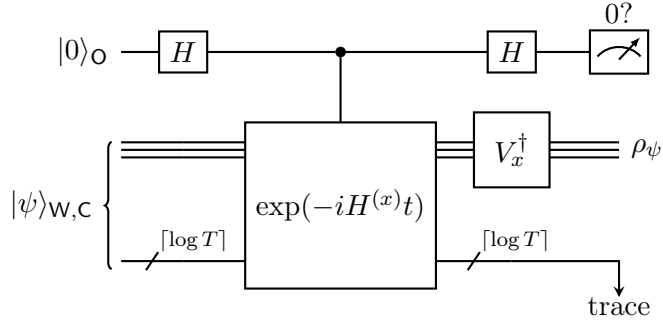
    
    Through direct calculation, we can express the acceptance probability as follows:
    \[p'_{\rm acc}\coloneqq \Pr{V'_x \text{ accepts } \ket{\psi}} = \frac{1}{4}\big\|V_x^{\dagger}(I+e^{-iH^{(x)}t}) \ket{\psi}\big\|_2^2 = \frac{1}{2}+\frac{1}{2}\operatorname{Re}(\bra{\psi} e^{-iH^{(x)}t} \ket{\psi}).\]
    We denote the resulting state by $\rho_{\psi}$ when $\exp(-iH^{(x)}t)$ is perfectly implemented, whereas the resulting state is $\rho'_{\psi}$ if the implementation error $\epsilon$ of $\exp(-iH^{(x)}t)$ is considered (e.g., \Cref{lemma:space-efficient-Hamiltonian-simulation}). 
    Simultaneously, $\rho_{\psi}$ is given by
    $\rho_{\psi}\coloneqq \frac{1}{4 p'_{\rm acc}} V_x^{\dagger} \Tr_{\sfC}\!\left(  (I+e^{-iH^{(x)}t}) \ket{\psi}\bra{\psi} (I+e^{iH^{(x)}t}) \right) V_x.$

    \vspace{1em}
    \noindent\textbf{Completeness}. We first prove that this \stateQMA{} verifier $V'_x$ meets the completeness condition. 
    Consider the unique ground state $\ket{\Omega_x}$ of the Hamiltonian $H^{(x)}$ specified in \Cref{lemma:Hamiltonian-from-UQMAverifier} as a witness state of $V'_x$, and let $\lambda_0\coloneqq \lambda_{\min}\big(H^{(x)}\big)=1-p_{\rm acc}=\nu$ be the ground energy of $H^{(x)}$ corresponding to the maximum acceptance probability $p_{\rm acc}$ of $V_x$. 
    Consequently, we obtain the following by combining with the implementation error $\epsilon$ of $\exp(-H^{(x)} t)$ specified in \Cref{lemma:space-efficient-Hamiltonian-simulation}: 
    \[\Pr{V'_x \text{ accepts } \ket{\Omega_x}} = \frac{1+\cos(\nu t)}{2} 
    \text{ and }
    c'(n) = \frac{1+\cos(\nu t)}{2}-\epsilon.\]

    \vspace{1em}
    \noindent\textbf{Soundness}. For the soundness condition, we consider a quantum state $\ket{\psi}=\alpha_0 \ket{\Omega_x} + \sum_{1 \leq j < N} \alpha_j \ket{\Phi_x^{(j)}}$, where the coefficients $\{\alpha_j\}_{0 \leq j < N}$ satisfies that $\sum_{0 \leq j < N} |\alpha_j|^2=1$ and $\ket{\Phi_x^{(j)}}$ is the $j$-th excited state of $H^{(x)}$, as a witness state of the \stateQMA verifier $V'_x$, we thus obtain the corresponding acceptance probability $p'_{\rm acc}$ and the resulting state $\rho_{\psi}$, respectively:
    
    \begin{equation}
        \label{eq:pcc-soundness-all}
        p'_{\rm acc} =\Pr{V'_x \text{ accepts } \ket{\psi}} = |\alpha_0|^2 \cdot \frac{1+\cos(\lambda_0)}{2} + \sum_{1 \leq j < N} |\alpha_j|^2 \cdot \frac{1+\cos(\lambda_j t)}{2} \geq s(n), 
    \end{equation}

    \begin{equation}
        \label{eq:resulting-state-soundness}
        \begin{aligned}
         \rho_{\psi} = \frac{1}{4 p'_{\rm acc}} V_x^{\dagger} \Tr_{\sfC} \Bigg(&\Big(\alpha_0(1+e^{-i\lambda_0 t}) \ket{\Omega_x} + \sum_{1\leq j < N} \alpha_j (1+e^{-i\lambda_j t}) \ket{\Phi_x^{(j)}}\Big) \\
        &\Big( \alpha^*_0(1+e^{i\lambda_0 t}) \bra{\Omega_x} + \sum_{1\leq j < N} \alpha^*_j (1+e^{i\lambda_j t}) \bra{\Phi_x^{(j)}} \Big) \Bigg) V_x.
        \end{aligned}
    \end{equation}

    Here, $\lambda_j$ (${1 \leq j < N}$) is the $j$-th excited\footnote{It is noteworthy that the $j$-th excited energy of a Hamiltonian corresponds to the $(j+1)$-th smallest eigenvalue of $H$, where $H$ is viewed as an Hermitian matrix. Similarly, the $j$-th excited state of a Hamiltonian $H$ corresponds to the eigenvector associated with the $(j+1)$-th smallest eigenvalue of $H$.} energy corresponding to $j$-th excited state $\ket{\Phi_x^{(j)}}$ such that $\lambda_1 \leq \cdots \leq \lambda_{N-1} \leq \pi$. 
    It is noteworthy that the fidelity (overlap) between the resulting state $\rho_{\psi}$ and the unique witness state $\ket{w_x}\bra{w_x}$ is mostly decided by $|\alpha_0|^2$. 
    To upper-bound $\td(\rho_{\psi},\ket{w_x}\bra{w_x})$, it suffices to minimize $|\alpha_0|^2$ amongst all quantum states $\ket{\psi}$ such that $V'_x$ accepts $\ket{\psi}$ with probability at least $s(n)$. 
    Accordingly, we need to maximize $|\alpha_1|^2$ because $\exp(-i\lambda_1 t)$ contributes the largest possible weight amongst all eigenvalues of $H^{(x)}$ except for $\lambda_0$. Meanwhile, we can without loss of generality assume that $\alpha_2=\cdots=\alpha_{N-1}=0$, then we obtain the following derived from \Cref{eq:pcc-soundness-all}: 
    \begin{equation}
        \label{eq:pcc-soundness}
        p'_{\rm acc} = |\alpha_0|^2 \cdot \frac{1+\cos(\lambda_0 t)}{2} + (1-|\alpha_0|^2)\cdot \frac{1+\cos(\lambda_1 t)}{2} \geq s'(n).
    \end{equation}
    
    Furthermore, we can assume that $(1+\cos(\lambda_1 t)/2) \leq s'(n) \leq (1+\cos(\lambda_0 t))/2$, otherwise $\ket{\psi}$ will be orthogonal to $\ket{\Omega_x}$ which means that $|\alpha_0|^2=0$. By adapting $\rho_{\psi}$ in \Cref{eq:resulting-state-soundness} according to the choice of $\{\alpha_i\}_{2\leq i < N}$ with $\lambda_0=\nu$, we have derive that: 
    \begin{equation}
    \label{eq:fidelity-lower-bound}
    \begin{aligned}
        \F^2(\rho_{\psi},\ket{w_x}\bra{w_x}) &= \bra{w_x} \rho_{\psi} \ket{w_x}\\
        &\geq \frac{1}{4 p'_{\rm acc}} \bra{w_x} V_x^{\dagger} \Tr_{\sfC} \big( |\alpha_0|^2 (1+e^{-i\nu t}) \ket{\Omega_x}\bra{\Omega_x} (1+e^{i\nu t}) \big) V_x \ket{w_x}\\
        & \geq \frac{|\alpha_0|^2}{p'_{\rm acc}} \cdot \frac{(1+\cos(\nu t))^2}{4} \cdot (1-\delta_1)\\
        & \coloneqq  \frac{|\alpha_0|^2}{p'_{\rm acc}} \cdot (1-\delta_0)^2 \cdot (1-\delta_1)
    \end{aligned}
    \end{equation}
    
    Here, the third line is because of $\bra{w_x} V_x^{\dagger} \Tr_{\sfC} \big(  \ket{\Omega_x}\bra{\Omega_x} \big) V_x \ket{w_x} \geq 1-\delta_1$, and the last line defines $\delta_0\coloneqq (1-\cos(\nu t))/2$.
    By employing \Cref{lemma:td-vs-fidelity} on \Cref{eq:fidelity-lower-bound}, we obtain: 
    \begin{equation}
        \label{eq:trace-dist-upper-bound}
        \begin{aligned}
        \td(\rho_{\psi},\ket{w_x}\bra{w_x}) &\leq 1-\F(\rho_{\psi},\ket{w_x}\bra{w_x})
        = 1-(1-\delta_0)\cdot \sqrt{1-\delta_1} \cdot \frac{|\alpha_0|}{p'_{\rm acc}}.
        \end{aligned}
    \end{equation}
    We then establish a lower bound of $|\alpha_0|^2/p'_{\rm acc}$ and the corresponding soundness parameter $s'$ in \Cref{prop:UQMAwitness-bound}, and the proof will be deferred. 
    \begin{proposition}
        \label{prop:UQMAwitness-bound}
        $|\alpha_0|^2/p'_{\rm acc} \geq 1-q(\nu\delta_1 t/T^3)$ for some polynomial $q$. 
        Moreover, the choice of parameters implies that $s'(n) = \frac{1}{4} \big((1+\cos(\lambda_1 t)\big)\big(2+\cos(\nu t) - \cos(\lambda_1 t)\big) + \epsilon$. 
    \end{proposition}
    
    Plugging \Cref{prop:UQMAwitness-bound} into \Cref{eq:trace-dist-upper-bound} with  $\delta'\coloneqq  \max\{\delta_0,\delta_1,q(\nu\delta_1 t/T^3)\}$, we obtain: 
    \[\td(\rho_{\psi}, \ket{w_x}\bra{w_x}) \leq 1-(1-\delta_0)\sqrt{(1-\delta_1)(1-q(\nu\delta_1 t/T^3))} \leq 1-(1-\delta')^2 = 2\delta'-(\delta')^2 \leq 2\delta'.\]
    By utilizing \Cref{lemma:space-efficient-Hamiltonian-simulation}, we note that the implementation error $\epsilon$ of $\exp(-H^{(x)}t)$ is exponentially small, and thus finish the proof by noticing $\td(\rho'_{\psi}, \ket{w_x}\bra{w_x}) \leq O(\delta')$.     
\end{proof}

At the end of this subsection, we provide the proof of \Cref{prop:UQMAwitness-bound}. 

\begin{proof}[Proof of \Cref{prop:UQMAwitness-bound}]
    We establish a lower bound for $|\alpha_0|^2/p'_{\rm acc}$ by a direct calculation: 
    \begin{equation*}
        \label{eq:UQMAwitness-bound1}
        \begin{aligned}
        \frac{|\alpha_0^2|}{p'_{\rm acc}} &= \frac{|\alpha_0|^2}{|\alpha_0|^2 \cdot \frac{1+\cos(\lambda_0 t)}{2} + (1-|\alpha_0|^2)\cdot \frac{1+\cos(\lambda_1 t)}{2}} ~{\coloneqq  f(|\alpha_0|^2)}\\
        &\geq \frac{(2s(n)-1) - \cos(\lambda_1 t)}{s(n) \cdot (\cos(\lambda_0 t) - \cos(\lambda_1 t))}\\
        &= \frac{2\xi}{1+\xi \cos(\lambda_0 t) + (1-\xi) \cos(\lambda_1 t)} ~{\coloneqq  g(\xi)}\\
        &\geq \frac{2}{1+\cos(\lambda_1 t)} \cdot \xi - \frac{2(\cos(\lambda_0 t)-\cos(\lambda_1 t))}{(1+\cos(\lambda_1 t)^2)} \cdot \xi^2. 
        \end{aligned}
    \end{equation*}
    Here, the second line is due to the fact that the function $f(|\alpha_0|^2)$ is monotonically increasing on the interval $[0,1]$, and the minimum achieves when $|\alpha_0|^2 = \frac{(2s(n)-1) - \cos(\lambda_1 t)}{\cos(\lambda_0 t) - \cos(\lambda_1 t)}$ following from \Cref{eq:pcc-soundness}. The third line owing to choose $s\coloneqq \xi \cdot \frac{1+\cos(\lambda_0)}{2} + (1-\xi) \cdot \frac{1+\cos(\lambda_1 t)}{2}$ for $0 \leq \xi \leq 1$, and the last line is because of truncated Taylor series of the function $g(\xi)$. 

    By choosing $\xi = (1+\cos(\lambda_1 t))/2$ and utilizing the bound $\Delta\big(H^{(x)}\big)=\lambda_1 - \lambda_0 = \lambda_1 - \nu \geq \Omega(\nu\delta_1/T^3)$ in \Cref{lemma:Hamiltonian-from-UQMAverifier}, we have the following for some polynomial $q$: 
    \begin{equation*}
        \label{eq:UQMAwitness-bound2}
         \frac{2}{1+\cos(\lambda_1 t)} \cdot \xi - \frac{2(\cos(\nu t)-\cos(\lambda_1 t))}{(1+\cos(\lambda_1 t)^2)} \cdot \xi^2 = 1- \frac{\cos(\nu t) - \cos(\lambda_1 t)}{2} \geq 1-q\Big( \frac{\nu\delta_1 t}{T^3} \Big). 
    \end{equation*}

    Accordingly, we obtain the soundness parameter $s$ by combining the implementation error $\epsilon$ of $\exp(-iH^{(x)}t)$ specified in \Cref{lemma:space-efficient-Hamiltonian-simulation}: 
    $s'(n) = \big((1+\cos(\lambda_1 t)\big)\big(2+\cos(\nu t) - \cos(\lambda_1 t)\big)/4 + \epsilon.$
\end{proof}

\subsection{Weighted circuit-to-Hamiltonian construction, revisited}
\label{subsec:weighted-FK-construction}
Finally, we demonstrate the weighted circuit-to-Hamiltonian construction with a spectral gap lower bound, along the line of~\cite{BC18,BCNY19}. 

\begin{proof}[Proof of \Cref{lemma:Hamiltonian-from-UQMAverifier}]
    We consider an \UQMA{} verifier $V_x$ with maximum acceptance probability $1-\nu$, where the verification circuit consists of a sequence of local quantum gates $U_1,\cdots,U_T$. To begin, we specify a weighted circuit-to-Hamiltonian construction $H^{(x)}\coloneqq H^{(x)}_{\rm prop} + H_{\rm in} + H_{\rm out}$ proposed in~\cite{BC18}.\footnote{We note that by utilizing a spacetime circuit-to-Hamiltonian construction outlined in~\cite{BCNY19}, we can achieve a constantly improved distance parameter $\delta$.} Here, $H_{\rm in}\coloneqq \ket{\bar{0}}\bra{\bar{0}} \otimes \Pi_{\rm in}$ and $H_{\rm out}\coloneqq \ket{T}\bra{T} \otimes \Pi_{\rm out}$ correspond to input and output constraints, respectively. It is straightforward to see that the ground energy of $H^{(x)}$ is $1-(1-\nu)=\nu$. This construction closely resembles the standard circuit-to-Hamiltonian construction~\cite{KSV02}, except for the weighted ground state resulting from the distinct propagation term $H^{(x)}_{\rm prop}$.

    \vspace{0.75em}
    \noindent\textbf{Constructing the propagation term.} 
    Next, we proceed to construct the propagation term $H^{(x)}{\rm prop}$ for the given weights, specifically a probability distribution $\pi$. We define the ground state of $H^{(x)}$ as $\ket{\Omega_x} \coloneqq  \sum_{t=0}^T \sqrt{\pi_t} \ket{t} \otimes (U_t\cdots U_1) \ket{w_x}$, where $\ket{w_x}$ represents the unique witness state of $V_x$, and $\pi=(\delta_1/T,\cdots,\delta_1/T,1-\delta_1)$ is a probability distribution on $[T]=\{0,1,\cdots,T\}$. 
    
    Our construction of the propagation term is based on the transition matrix $P$ of random walks on a $(T+1)$-node path with the stationary distribution $\pi$. Here, the specific form of the matrix $P$ is provided in \Cref{lemma:RW-on-path} for the given $\pi$. By employing \Cref{lemma:RW-on-path}, we obtain an $O(\log n)$-local stoquastic frustration-free Hamiltonian $H^{(T)}_{\rm path}$ derived from $P$. Let us consider a unitary operator $W_x\coloneqq \sum_{t=0}^{T} \ket{t}\bra{t}\otimes U_t\cdots U_1$ associated with the given $V_x$. By combining $H^{(T)}_{\rm path}$ with Lemma 2 in~\cite{BC18}, we arrive at the expression for the propagation term $H^{(x)}_{\rm prop}$: 
    \begin{align*}
    H^{(x)}_{\rm prop} &= W_x (H^{(T)}_{\rm path} \otimes I) W_x^{\dagger} \\
    &=\sum_{t=0}^T (1-P_{t,t})\ket{t}\bra{t}\otimes I-\sum_{t=0}^T \left(\tfrac{\sqrt{\pi_t}}{\sqrt{\pi_{t-1}}}P_{t,t-1} \ket{t}\bra{t-1}\otimes U_t + \tfrac{\sqrt{\pi_{t-1}}}{\sqrt{\pi_t}} P_{t-1,t} \ket{t-1}\bra{t}\otimes U_t^{\dagger} \right).
    \end{align*}

    \vspace{0.75em}
    \noindent\textbf{Lower-bounding the spectral gap.} 
    It is left to establish a lower bound for the spectral gap $\Delta(H^{(x)})$. 
    By applying Cheeger's inequality (\Cref{lemma:Cheeger-inequality}), we obtain $\Delta(H^{(x)}_{\rm prop}) = \Delta_P \geq \frac{1}{2} \Phi_G^2$, where $\Phi_G$ is the conductance of the underlying $(T+1)$-node path graph $G$. 
    Notably, the weights in $G$ are primarily concentrated on the $(T+1)$-th node since $\delta_1$ is tiny. 
    The minimum conductance $\Phi_G=\min_{S \subseteq [T]} \Phi_G(S)$ is thus achieved at the subset $S=[T-1]\coloneqq \{0,1,\cdots,T-1\}$. Consequently, we can deduce: 
    \begin{equation}
        \label{eq:conductance-lower-bound}
        \Phi_G = \Phi_G([T-1]) = \frac{1}{\pi_{[T-1]}} \pi_{T-1} P_{T-1,T} = \frac{1}{\delta_1} \cdot \pi_{T-1} \cdot \frac{1}{4}\min\left\{1,\frac{\pi_T}{\pi_{T-1}}\right\} = \frac{1}{4T}.
    \end{equation}
    From \Cref{eq:conductance-lower-bound}, we infer that $\Delta(H^{(x)}{\rm prop}) \geq \Omega(T^{-2})$.
    
    Now, let $C(\tau,T)$ be the set of quantum circuits of size $T$ that achieve the maximum acceptance probability $\tau$ amongst states satisfying the input and output constraints $\Pi_{\In}$ and $\Pi_{\Out}$, respectively. In particular, $C(\tau,T)\coloneqq \big\{U_1,\cdots,U_T: \max\limits_{\ket{\xi}\in\ker(\Pi_{\In}),\ket{\eta}\in\ker(\Pi_{\Out})} |\bra{\xi} U_1\cdots U_T \ket{\eta}|^2 =\tau\big\}$. By employing Lemma 3 in~\cite{BC18}, we obtain the following when $\tau=1-\nu$:
    \begin{equation}
    \label{eq:spectral-gap-lower-bound}
    \begin{aligned}
    \Delta(H^{(x)}) &\geq \min_{U\coloneqq U_1,\cdots, U_T \in C(1-\nu,T)} \left( \lambda_{\min}\big(H^{(U)}\big) - \lambda_{\min} \big(H^{(U)}_{\rm prop}\big) \right)\\
    &\geq \Delta\big(H_{\rm path}^{(T)}\big) \cdot \frac{1-\sqrt{1-\nu}}{4} \cdot \min\{\pi_0,\pi_T\} \\
    &= \Omega\left( \frac{\nu\delta_1}{T^3} \right).
    \end{aligned}
    \end{equation}
    Here, the last line is because $\sqrt{1-x} \leq 1-x/2$ when $0 \leq x \leq 1$. 
    As \Cref{eq:spectral-gap-lower-bound} establishes the desired spectral gap lower bound, this thus finishes the proof. 
\end{proof} 


\section{\stateQCMA{} achieves perfect completeness}
\label{sec:stateQCMA}

In this section, we will present a state-synthesizing analogue of the $\QCMA=\QCMA_1$ theorem~\cite{JKNN12}. 

\begin{restatable}[\stateQCMA{} is closed under perfect completeness]{theorem}{stateQCMAPerfectCompleteness}
\label{thm:stateQCMA-perfect-completeness}
For any efficiently computable functions $c(n)$, $s(n)$ and $\delta(n)$ such that $c(n)-s(n) \geq 1/\poly(n)$, we have that 
\[\stateQCMA_{\delta}[c,s] \subseteq \stateQCMA_{\delta'}[1,s'],\]
where $s' = \tfrac{1}{2}\left( \tfrac{s}{c}\right)^3 - 2\left( \tfrac{s}{c}\right)^2 + \tfrac{5}{2}\left( \tfrac{s}{c}\right)$ and $\delta'(n)=\delta(n)+\exp(-\poly(n))$.

Furthermore, for any \stateQCMA{} verifier utilizing the ``Pythagorean'' gateset, we have $\delta'=\delta$. 
\end{restatable}

It is noteworthy that while error reduction for \stateQCMA{} (a corollary of \Cref{thm:error-reduction-stateQMA}) preserves the distance between the resulting state and the target state, this distance-preserving property in \Cref{thm:stateQCMA-perfect-completeness} is \textit{gateset-dependent} since changing gatesets will worsen the distance parameter (\Cref{lemma:changing-gateset-errs}). This is because the key insight from~\cite{JKNN12}, that the maximum acceptance probability of a \QCMA (or \stateQCMA{}) verifier can be expressed with polynomially many bits, only applies to certain gatesets, elaborating in \Cref{remark:choices-of-gatesets}.   

\begin{remark}[On choices of the gateset]
\label{remark:choices-of-gatesets}
For state-synthesizing complexity classes, we require to deal with not only the maximum acceptance probability but also the resulting states after the computation. To well-approximate any quantum states with a designated gateset $\calS$, this gateset $S$ must generate a dense subgroup of $\mathrm{SU}(2^n)$ for large enough $n$. However, this does not hold for the gateset used in~\cite{JKNN12}.\footnote{To be specific, the proof of Theorem 3.2 in~\cite{Shi02} indicates that the gateset consists of \Toffoli{} and \Hadamard{} merely generates a dense subgroup of $\mathrm{SO}(8)$.} 
We then use the ``Pythagorean'' gateset in~\cite{JN11} where real and imaginary parts of all matrix entries are rational numbers. 
\end{remark}

Analogous to~\cite{JKNN12}, to achieve perfect completeness, we also utilize the quantum rewinding lemma~\cite{Wat09} with a single iteration. As stated in \Cref{lemma:one-iteration-quantum-rewinding}, it suffices to construct a new \stateQCMA{} verifier that the acceptance probability is exactly $1/2$ for the completness condition.  

Our construction of the new $\stateQCMA$ verifier with perfect completeness tightly follows the construction for \QCMA{} (i.e., Figure 1 in~\cite{JKNN12}). The only difference is that the unitary transformation $Q$ since the original construction in~\cite{JKNN12} cannot preserve the resulting states, which also leads to a slightly different analysis for the soundness condition. 

\vspace{1.5em}
We begin with two lemmas that are crucial for our analysis. 

\begin{lemma}[Acceptance probabilities from the ``Pythagorean'' gateset are rational, adapted from \cite{JN11}]
\label{lemma:pythagorean-gateset-is-rational}
For all unitary transform $U$ on an $n$-qubit system that consists of $l(n)$ gates from the ``Pythagorean'' gateset where $l$ is a polynomial, the probability $p_{acc}$ that the first qubit of $U\ket{0^n}$ is found in the state $\ket{1}$ (i.e., the measurement on the computational basis) is expressed as $p_{acc}=\frac{k}{5^{l(n)}}$ where $k$ is an integer from the range $[0,5^{l(n)}]$.
\end{lemma}

\begin{lemma}[Success probability of the quantum rewinding, adapted from Lemma 8 in~\cite{Wat09}]
\label{lemma:one-iteration-quantum-rewinding}
Let $Q$ be a quantum circuit such that $Q\ket{\psi}=\sqrt{p(\psi)}\ket{0}\ket{\psi_0}+\sqrt{1-p(\psi)}\ket{1}\ket{\psi_1}$, and $p\coloneqq p(\psi)\in (0,1)$ is constant over all choices of the input $\ket{\psi}$. 
Then the probability of producing $\ket{\psi_0}$ within $t$ iterations is $1-(1-p)(1-2p)^{2t}$. 

Particularly, the probability of producing $\ket{\psi_0}$ by the single-iteration quantum rewinding of $Q$ is $p+4p(1-p)^2$, as well as this achieves the perfect success probability when $p=1/2$.
\end{lemma}

Now we proceed with the actual proof. 

\begin{proof}[Proof of \Cref{thm:stateQCMA-perfect-completeness}] 
Let $V_{x,w}$ be a $\stateQCMA_{\delta}[c,s]$ verification circuit with a classical witness $w$ of size $m(n)$ where $n\coloneqq |x|$ and $c(n)$ is a sufficiently large constant.\footnote{This is achievable by first applying error reduction on the given \stateQCMA{} verifier.} Using Lemma~\ref{lemma:changing-gateset-errs}, we first convert $V_{x,w}$ to another circuit that utilizes a polynomial number of gates from the gateset $\mathcal{G}$. Note that while we can assume that the completeness parameter $c(n)$ does not increase by this conversion, the trace distance error increases to $\delta'(n)\coloneqq \delta(n) + \exp(-\poly(n))$ owing to the triangle inequality. 
In the rest of the proof, we use $V_{x,w}$ as the converted circuit. So $V_{x,w}$ is assumed to be a $\stateQCMA_{\delta'}[c,s]$ circuit. Utilizing Lemma~\ref{lemma:pythagorean-gateset-is-rational}, the acceptance probability of each $V_{x,w}$ is expressed as $k_{x,w}/5^{l(n)}$ for some integer $k_{x,w}\in\{0,1,\cdots,5^{l(n)}\}$ and polynomial $l$ which represents the size of the circuit $V_{x,w}$.

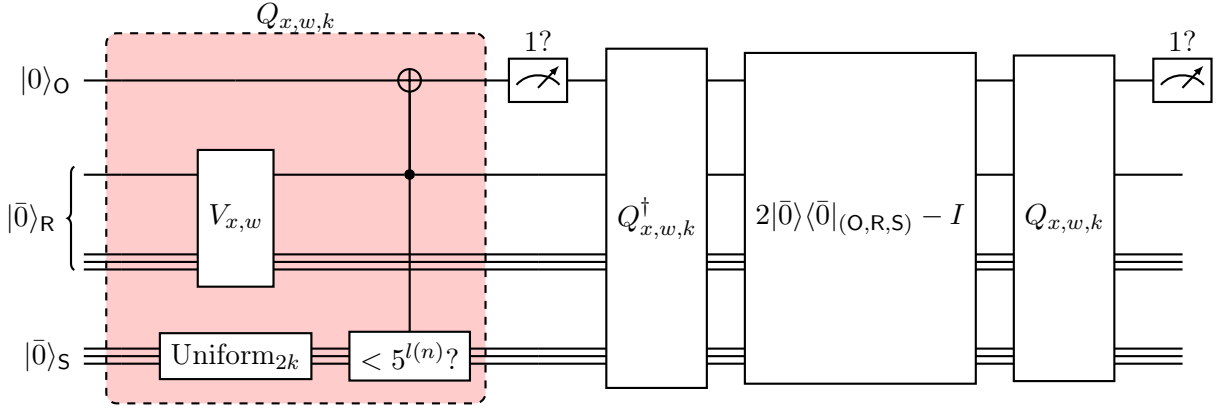
\begin{figure}[ht!]
\centering
\begin{quantikz}[wire types={q,q,b,b}, classical gap=0.07cm]
	\lstick{$\ket{0}_{\sfO}$} & \qw\gategroup[4,steps=3,style={dashed, rounded corners,fill=red!20, inner xsep=2pt, inner ysep=2pt}, background]{$Q_{x,w,k}$}  & \qw & \targ{} & \meter{1?} & \gate[4]{Q^{\dagger}_{x,w,k}} & \gate[4]{2\ket{\bar{0}}\bra{\bar{0}}_{(\sfO,\sfR,\sfS)}-I} & \gate[4]{Q_{x,w,k}} & \meter{1?}\\
	\lstick[wires=2]{$\ket{\bar{0}}_{\sfR}$} & & \gate[2]{V_{x,w}} & \ctrl{-1} & & & & & \\
		& & & & & & & & \\
 	\lstick{$\ket{\bar{0}}_{\sfS}$} & & \gate[1]{\text{Uniform}_{2k}} & \gate[1]{< 5^{l(n)}?}\vqw{-3} & & & & &
\end{quantikz}
\caption{The new verification circuit $\tilde{V}_{x,w}$}	
\label{fig:stateQCMA-new-verifier}
\end{figure}

Now we construct a $\stateQCMA_{\delta'}[1,s']$ verifier such that the verification circuit $\tilde{V}_{x,w}$ is the single-iteration quantum rewinding of $Q_{x,w,k}$ with a pre-processing, where the unitary transformation $Q_{x,w,k}$ is shown in \Cref{algo:Q}. 
To be precise, $\tilde{V}_{x,w}$ will simply reject if $k/5^{l(n)} < c(n)$,\footnote{The claimed maximum acceptance probability $k/5^{l(n)}$ are supposed to correspond to \textit{legal} witnesses.} then $\tilde{V}_{x,w}$ will perform the construction in \Cref{fig:stateQCMA-new-verifier}. 

	\begin{algorithm}[ht!]
        1. Prepare the uniform superposition $\frac{1}{\sqrt{2k}}\sum_{z\in\{0,1,\cdots,2k-1\}} \ket{z}$ in $\sfS$ \;
        2. Apply $V_{x,w}$ on $\sfR$\;
        3. Apply a \Toffoli{} gate that targets at $\sfO$, and the first control qubit is the designated output qubit in $\sfR$, as well as the second control qubit is decided by whether the integer in $\sfS$ at most $5^{l(n)}$. 
        \BlankLine
		\caption{Unitary Transformation $Q_{x,w,k}$}
		\label[algorithm]{algo:Q}
	\end{algorithm}

In \Cref{algo:Q}, all registers $\sfO,\sfR,\sfS$ are assumed to be initialized to all-zero state, where $\sfO$ is a single-qubit register, and $\sfR$ is the $(q(n) + n)$-qubit register for some polymonial $q$, on which the verification circuit $V_{x,w}$ acts (and the last $n$-qubit of $\sfR$ is assumed to contain the output), as well as $\sfS$ is a $t$-qubit register where $t$ is the minimum integer satisfying $2^t \geq 5^{l(n)}$. 
Moreover, the first step in \Cref{algo:Q} can be implemented by exact amplitude amplification.\footnote{Note that any integer in $[0,5^{l(n)}-1]$ can be decomposed into a product of some power of $2$ and some odd integer. Then we can employ with the construction in Figure 12 in~\cite{BGB+18}.}
We then show that the unitary transformation $Q_{x,w,k}$ indeed preserves the resulting states of the verification circuit $V_{x,w}$ as \Cref{prop:synthesized-state-preserving-Q}, whose proof is deferred to the last part of this section.
\begin{proposition}
\label{prop:synthesized-state-preserving-Q}
For the unitary transformation $Q_{x,w,k}$ specified in \Cref{algo:Q}, the success probability is $\Pr{Q_{x,w,k} \text{ accepts}}\coloneqq \| \Pi_{\rm acc} Q_{x,w,k} \ket{\bar{0}}_{(\sfO,\sfR,\sfS)} \|^2 = k_{x,w}/(2k)$ where $\Pi_{\rm acc}\coloneqq \ket{1}\bra{1}_{\sfO}\otimes I_{(\sfR,\sfS)}$. Moreover, $Q_{x,w,k}$ preserves the state synthesized by $V_{x,w}$. 
\end{proposition}

Plugging \Cref{prop:synthesized-state-preserving-Q} into \Cref{lemma:one-iteration-quantum-rewinding}, we obtain that
\begin{equation}
\label{eq:stateQCMA-Pacc}
\Pr{\tilde{V}_{x,w} \text{ accepts }} = \frac{1}{2}\left( \frac{k_{x,w}}{k}\right)^3 - 2\left( \frac{k_{x,w}}{k}\right)^2 + \frac{5}{2}\left( \frac{k_{x,w}}{k}\right)\coloneqq g\left(\frac{k_{x,w}}{k}\right),
\end{equation}
as guaranteed that $k \geq c(n)\cdot 2^{l(n)}$. We further notice the following fact:\footnote{\Cref{fact:stateQCMA-Pacc-mono} follows from the facts that $t=1$ and $t=5/3$ are roots of $g'(t)=0$, as well as $g'(0) > 0$.} 
\begin{fact}
\label{fact:stateQCMA-Pacc-mono}
For $0 \leq t \leq 1$, $g(t)$ defined in \Cref{eq:stateQCMA-Pacc} is monotonically increasing.
\end{fact}
It is thus left to analyze the maximum acceptance probability of $\tilde{V}_{x,w}$. 

\begin{itemize}
    \item For the completeness condition, equipped with \Cref{fact:stateQCMA-Pacc-mono}, the maximum achieves if $t\coloneqq k_{x,w}/k=1$. Namely, the claimed maximum acceptance probability $k$ of $V_{x,w}$ is indeed the maximum acceptance probability $k_{x,w}$. 
    \item For the soundness condition, we observe below for any illegal witness $w\in\binset^{m(n)}$:
    \[\frac{k_{x,w}}{k} \leq \frac{s(n)\cdot 5^{l(n)}}{c(n)\cdot 5^{l(n)}} = \frac{s(n)}{c(n)}.\]
    Since $k_{x,w}/k\leq s(n)/c(n)$, 
    this indicates that the acceptance probability of $V_{x,w}$ is $s'(n)\leq g(s(n)/c(n))$. By \Cref{fact:stateQCMA-Pacc-mono}, $\Pr{V_{x,w} \text{ accepts }}\leq s'(n)$ holds for any choice of $k$ as long as the witness $w$ is illegal. 
    Note that $c$ is a constant and the promise gap $c(n)-s(n)\geq 1/p(n)$ where $p(n)$ is a polynomial of $n$, this deduces that 
    \begin{equation}
        \label{eq:stateQCMA1-soundness}
        s'(n)\leq g\left(1-\frac{c(n)}{p(n)}\right)=1-\frac{1}{2}\left[\left(\frac{c(n)}{p(n)}\right)^2+\left(\frac{c(n)}{p(n)}\right)^3\right]\coloneqq 1-\frac{1}{q(n)}.
    \end{equation}
\end{itemize}
    We complete the proof by noticing $q(n)$ is a polynomial of $n$. 
\end{proof}

In the remaining of this section, we complete the proof of \Cref{prop:synthesized-state-preserving-Q}.

\begin{proof}[Proof of \Cref{prop:synthesized-state-preserving-Q}]
Let $V_{x,w}\ket{\bar{0}}_{\sfR} = \ket{0}\ket{\phi_0} + \ket{1}\ket{\phi_1}$ be the quantum state just before the final measurement of the verification circuit $V_{x,w}$ where $n\coloneqq |x|$. This deduces that 
\[\Pr{V_{x,w} \text{ accepts }}=\|\ket{1}\bra{1}_{\Out}V_{x,w}\ket{\bar{0}}_{\sfR}\|_2^2=\innerprod{\phi_1}{\phi_1}=k_{x,w}/5^{l(n)}.\]
The output state conditioned on accepting is then defined as the state we obtain by tracing out all but the last $n$-qubit of the density matrix $\frac{5^{l(n)}}{k_{x,w}}\ket{\phi_1}\bra{\phi_1}$.
The quantum state in $(\mathsf{O,R,S})$ before the last step in \Cref{algo:Q} is then
\begin{equation*}
\ket{0}_{\mathsf{O}}\otimes \big[ \ket{0}\ket{\phi_0}+\ket{1}\ket{\phi_1} \big]_{\mathsf{R}}\otimes \frac{1}{\sqrt{2k}} \sum_{0\leq z <2k} \ket{z}_{\mathsf{S}}
\end{equation*}
and therefore the state after applying \Cref{algo:Q} is $Q_{x,w,k}\ket{\bar{0}}_{(\mathsf{O,R,S})}$ as below: 
\[\ket{0}_{\mathsf{O}}\otimes \ket{0}\ket{\phi_0}_{\mathsf{R}} \otimes \frac{1}{\sqrt{2k}}\sum_{z=0}^{2k-1} \ket{z}_{\mathsf{S}}
+\ket{0}_{\mathsf{O}}\otimes \ket{1}\ket{\phi_1}_{\mathsf{R}} \otimes \frac{1}{\sqrt{2k}}\sum_{z=5^{l(n)}}^{2k-1} \ket{z}_{\mathsf{S}} 
+\ket{1}_{\mathsf{O}}\otimes \ket{1}\ket{\phi_1}_{\mathsf{R}} \otimes \frac{1}{\sqrt{2k}}\sum_{z=0}^{5^{l(n)}-1} \ket{z}_{\mathsf{S}}\]
We thus conclude that
\begin{align*}
\Pr{Q_{x,w,k} \text{ accepts}} 
= \innerprod{\phi_1}{\phi_1} \cdot \bigg\|\frac{1}{\sqrt{2k}}\sum_{z=0}^{5^{l(n)}-1} \ket{z}_{\mathsf{S}}\bigg\|^2_2 = \frac{k_{x,w}}{5^{l(n)}} \cdot \frac{5^{l(n)}}{2k} = \frac{k_{x,w}}{2k}.
\end{align*}
Furthermore, we notice that the resulting state of $Q_{x,w,k}$ is 
\[(\ket{1}\bra{1}\otimes(\ket{\phi_1}\bra{\phi_1})_{\sfR}\otimes \left(\ket{0}\bra{0}\otimes \ket{+}\bra{+}^{\otimes l(n)}\right)_{\sfS}.\]
By tracing out $\mathsf{S}$ and all but the last $n$-qubit of $\mathsf{R}$, we get the resulting state of $V_{x,w}$. 
\end{proof}


\section*{Acknowledgments}
\noindent
We are grateful to Henry Yuen and Chinmay Nirkhe for insightful discussions and for bringing~\cite{BCNY19} to our attention. 
We also thank Gregory Rosenthal for suggesting the motivation behind defining the quantum state-synthesizing classes by drawing a comparison with classical function classes. 
Additionally, YL thanks Naixu Guo for helpful discussions on quantum singular value transformations. 
Furthermore, we express our gratitude to anonymous reviewers for providing detailed and helpful comments.

The authors were supported by JSPS KAKENHI Grants Nos.\ JP19H04066, JP20H00579, JP20H05966, JP20H04139, JP21H04879 and MEXT Quantum Leap Flagship Program (MEXT Q-LEAP) Grants Nos.\ JPMXS0118067394 and JPMXS0120319794. MM was also supported by JST, the establishment of University fellowships towards the creation of science technology innovation, Grant No. JPMJFS2120. 
Circuit diagrams were drawn by the Quantikz package~\cite{Kay18}. 

\bibliographystyle{alphaurlQ}
\bibliography{stateQMA.bib}

\end{document}